\documentclass[aps,pra,showpacs,twoside,twocolumn,nofootinbib,10pt]{revtex4-2}
\usepackage[colorlinks=true, citecolor=red, urlcolor=blue ]{hyperref}
\usepackage{times,epsfig,amssymb,amsfonts,amsmath,bm,subfigure,mathtools,amsthm,braket,soul,enumitem,color,physics,graphics,graphicx}
\usepackage[normalem]{ulem}
\usepackage{comment}
\usepackage{epstopdf}
\usepackage{multirow}
\usepackage{booktabs}
\usepackage{algpseudocode}
\newtheorem{theorem}{Theorem}
\newtheorem{corollary}{Corollary}

\newtheorem{proposition}{Proposition}
\newtheorem{remark}{Remark}

\newtheorem{lemma}{Lemma}

\usepackage{optidef}
\usepackage{xurl}

\newcommand{\F}{\mathbb F}

\newcommand{\Z}{\mathbb Z}
\renewcommand{\Tr}{\operatorname{Tr}}
\newcommand{\conv}{\operatorname{conv}}
\newcommand{\Sp}{\operatorname{Sp}}
\newcommand{\ASp}{\operatorname{ASp}}
\newcommand{\AGL}{\operatorname{AGL}}

\newcommand{\rad}{\operatorname{rad}}
\newcommand{\BD}{\operatorname{BD}}
\newcommand{\SEP}{\operatorname{SEP}}

\newcommand{\APPT}{\operatorname{APPT}}
\newcommand{\PPPT}{\operatorname{PermPPT}}
\newcommand{\hex}[1]{\mathtt{#1}}

\begin{document}

\title{On the separability of Bell-diagonal states with absolutely positive partial transposition}

\author{Ayan Patra\(^{1}\), Beatrix C. Hiesmayr\(^{1,2}\)}
\affiliation{\(^1\)IT:U Interdisciplinary Transformation University, Freistädter Strasse 400, 4040 Linz, Austria\\
\(^2\)University of Vienna, Faculty of Physics, Währingerstrasse 17, 1090 Vienna, Austria}

\begin{abstract}
Absolutely positive partial transpose ($\APPT$) states, which remain PPT under every global unitary, provide a spectral relaxation of absolutely separable states, which remain separable under every such transformation. Despite substantial progress, whether these two notions coincide beyond qubit--qudit systems remains unresolved, motivating the study of structured state families in which the $\APPT$ spectral constraints and separability can be analyzed exactly. We study this problem for Bell-diagonal states in two settings: the standard two-qutrit Weyl–Heisenberg basis and the two-ququart lattice basis generated by two-qubit Pauli operators. In the two-qutrit case, we show that a linearized necessary condition for $\APPT$ defines a polytope of spectra in which every associated Weyl--Heisenberg Bell-diagonal state admits an exact decomposition into explicitly separable states. This yields a separability criterion that applies to a region even larger than the $\APPT$ set itself. In the two-ququart lattice Bell-diagonal setting, we show that a state remains PPT under every permutation of its Bell coefficients if and only if its six largest Bell coefficients sum to at most $1/2$, and we prove that all such states are separable. Equivalently, every entangled lattice Bell-diagonal state becomes NPT after a suitable coefficient permutation and is therefore not $\APPT$. The proofs combine convex-polytope analysis, affine and symplectic orbit reductions, and exact rational decompositions into separable states supported on affine Lagrangian subsets. Our results rule out entangled $\APPT$ states in both Bell-diagonal families, provide simple spectral inequalities as sufficient criteria for separability, and show that, within the lattice Bell-diagonal family, PPT under all Bell-coefficient permutations is already enough to guarantee separability. Being basis-specific, however, they do not by themselves settle whether $\APPT$ and absolute separability coincide.
\end{abstract}	

\maketitle

\section{Introduction}
\label{sec:introduction}
In the second quantum revolution, entanglement has become central to quantum information theory, fueling advantages in several quantum technologies, including communication~\cite{BennettWiesner1992,BennettEtAl1993}, computation~\cite{JozsaLinden2003,RaussendorfBriegel2001}, sensing~\cite{GiovannettiEtAl2004,GiovannettiLloydMaccone2011,PezzeEtAl2018}, and cryptography~\cite{Ekert1991,GisinEtAl2002}. Nevertheless, determining whether a general mixed bipartite state is separable is an NP-hard problem~\cite{Gurvits2003,Gharibian2010}. Although the positive-partial-transpose (PPT) criterion provides a necessary and sufficient condition for separability in \(2\otimes2\) and \(2\otimes3\) systems~\cite{Peres1996,Horodecki1996}, PPT-entangled states exist in higher dimensions~\cite{Horodecki_1997,Horodecki2009}.

A related spectral characterization was introduced by M. Ku\'s and K. \.Zyczkowski in their geometric study of entangled states~\cite{KusZyczkowski2001}: a state is called \emph{absolutely separable} if it remains separable under every global unitary transformation. In particular, for a $d\otimes d$ system, any state $\rho$ satisfying $\operatorname{Tr}\rho^2\le\frac{1}{d^2-1}$ is separable~\cite{karol1998,Gurvits2002}. Since this condition is unitarily invariant, any state satisfying it is separable under all global unitaries. Replacing separability with the PPT condition yields the notion of \emph{absolutely PPT} ($\APPT$) states: a state $\rho$ is $\APPT$ if $(U\rho U^\dagger)^{\Gamma_A}\geq0$ for every global unitary $U$. Both absolute separability and the $\APPT$ property depend only on the spectrum of $\rho$, and every absolutely separable state is necessarily $\APPT$.

Hildebrand~\cite{Hildebrand2007} obtained a complete characterization of $\APPT$ spectra in arbitrary bipartite dimensions in terms of a finite collection of linear matrix inequalities (LMIs) on the ordered eigenvalues. Johnston~\cite{Johnston2013} subsequently proved that absolute separability and the $\APPT$ property coincide in every $2\otimes d$ system, with $d\geq2$. Beyond the qubit--qudit case, Arunachalam, Johnston, and Russo established the equivalence of the $\APPT$ property and absolute separability for the isotropic-state family~\cite{ArunachalamJohnstonRusso2015}. Several works have investigated the limitations of standard entanglement criteria on $\APPT$ states and developed sufficient criteria for absolute separability
\cite{ArunachalamJohnstonRusso2015,AbellanetVidal2025}, while others have studied the geometry and extreme points of the sets of $\APPT$ and absolutely separable states~\cite{SongChen2025,AhiableKothakondaWinter2026}. Despite this progress, whether the two notions coincide remains an open problem. In other words, it is not known whether
\begin{equation*}
\mathrm{APPT}_{d_1,d_2}=\mathrm{AS}_{d_1,d_2}
\quad \text{for all } d_1,d_2\geq3,
\end{equation*}
where $\APPT_{d_1,d_2}$ and $\mathrm{AS}_{d_1,d_2}$ stand for the sets of absolutely PPT and absolutely separable states in $d_1\otimes d_2$, respectively. 

One way to approach this problem is to fix an orthonormal basis and determine whether every $\APPT$ state diagonal in that basis is separable. Such a basis-specific result does not by itself establish absolute separability, since an arbitrary global-unitary conjugate of a diagonal state need not remain diagonal in the chosen basis. It nevertheless provides a useful testing ground: an entangled $\APPT$ state diagonal in some basis would immediately disprove the conjectured equivalence, whereas a positive result would exclude $\APPT$ entanglement from a nontrivial family and might reveal structural reasons why spectral PPT conditions force separability.

Bell-diagonal states, which are convex mixtures of Bell states, are especially suitable for this approach. The mixing weights of the Bell states, which we call the Bell coefficients, are precisely the eigenvalues of the state. Therefore, spectral properties such as $\APPT$ can be analyzed directly in terms of these coefficients. In this work, we consider the Bell states generated by Weyl--Heisenberg and Pauli operators, which are standard unitary error bases for qudit and multiqubit systems. These operators provide a natural description of discrete quantum errors and underlie several quantum information theoretic tasks, including teleportation and dense coding~\cite{BennettWiesner1992,BennettEtAl1993,Werner2001}, stabilizer quantum error correction~\cite{Knill1996,Gottesman1997,Klappenecker2002}, and Clifford circuit formalism~\cite{Dehaene2003,Hostens2005}. Beyond this operational role, their commutation relations equip the Bell-state labels with a finite symplectic phase-space geometry. These structures have been used extensively in studies of the two-qutrit magic simplex and higher-dimensional lattice states~\cite{Baumgartner2006,Baumgartner2007,Popp2022,BenattiLiguori2007,BenattiKarbalaii2013}. Bell-diagonal families therefore combine a manageable spectral description, direct relevance to quantum information theoretic protocols, and a sufficiently rich entanglement structure.

In particular, we investigate whether entangled $\APPT$ Bell-diagonal states exist in two structurally distinct settings: the standard Weyl--Heisenberg Bell basis in $3\otimes3$~\cite{Baumgartner2006,Baumgartner2007,Popp2022} and the lattice Bell basis generated by two-qubit Pauli operators in $4\otimes4$~\cite{BenattiLiguori2007,BenattiKarbalaii2013}. The two-qutrit basis is labeled by the symplectic space $\mathbb F_3^2$, whereas the lattice basis is labeled by $\mathbb F_2^4$. These two constructions require separate analyses; in particular, the lattice basis is not equivalent under local unitaries to the standard ququart Weyl--Heisenberg basis. Note that the two-qutrit setting is already known to contain a substantial PPT-entangled region: a numerical study based on uniformly sampled Bell-diagonal states identified at least \(13.9\%\) of the Bell-diagonal PPT states as bound entangled, while the status of a further $5.1\%$ remained undetermined~\cite{Popp2022}. It is therefore nontrivial whether the stronger $\APPT$ requirement excludes entanglement from this family. A recent review on bipartite bound entanglement can be found in Ref.~\cite{Hiesmayr2025}.

For the two-qutrit Weyl--Heisenberg Bell simplex~\cite{Baumgartner2007}, we show that a linear inequality on the eigenvalues of density matrices serves as a sufficient condition for separability, and use it to show that every $\APPT$ state diagonal in the standard two-qutrit Weyl--Heisenberg Bell basis is separable. For the $4\otimes4$ lattice Bell simplex, we prove a stronger result, which we call permutation rigidity. We call a coefficient vector associated with a lattice Bell-diagonal state \emph{permutation-PPT} \((\PPPT)\) if the state remains PPT for every permutation of its sixteen Bell coefficients. We show that this holds if and only if the sum of the six largest coefficients does not exceed $1/2$. Thereafter, we show that every \(\PPPT\) lattice Bell-diagonal state is separable; equivalently, every entangled lattice Bell-diagonal state becomes NPT after a suitable permutation of its Bell coefficients. Since every permutation of the lattice Bell coefficients is implemented by a global unitary that permutes the Bell basis states, $\APPT$ implies \(\PPPT\). The preceding result therefore implies that all $\APPT$ lattice Bell-diagonal states in $4\otimes4$ are separable.

The paper is organized as follows. Section~\ref{sec:preli} presents the Weyl--Heisenberg and Pauli Bell bases, the notions of $\APPT$ and \(\PPPT\), the relevant affine symmetry groups. This section also includes an important lemma showing that certain Bell-diagonal states, defined through the symplectic phase-space geometry of the Bell-state labels, are separable. Section~\ref{sec:3x3results} proves the two-qutrit result through a spectral outer region, affine-orbit reduction, and exact separability certificates. Section~\ref{sec:4x4results} gives the spectral characterization of the \(\PPPT\) polytope for the $4\otimes4$ lattice Bell-diagonal states and proves the separability of this polytope. Finally, Section~\ref{sec:conclusion} presents concluding remarks. Detailed supporting arguments, symmetry analysis, and computational procedures are provided in the appendices.

\section{Preliminary}
\label{sec:preli}
In this section, we briefly review the Bell projectors in two-qutrit and two-ququart systems. For two-qutrit systems, we use the standard Weyl--Heisenberg unitaries, while for two-ququart systems, we identify each local ququart with two qubits and use the Pauli matrix algebra. We then state a lemma establishing the separability of a specific family of Bell-diagonal states, defined through the symplectic phase-space geometry of the Bell-state labels, which applies to both systems. Finally, we present the notions of absolutely PPT ($\APPT$) and permutation-PPT ($\PPPT$) states, together with their spectral properties.

\subsection{Bell projectors and twirling}
\label{subsec:bellprojectors_twril}
Let $\mathcal{J}$ be an index set with cardinality $\vert{}\mathcal{J}\vert{} = d$, and let $\{\ket{j} : j \in \mathcal{J}\}$ denote an orthonormal basis for $\mathbb{C}^d$. The reference maximally entangled state (Bell state) in $\mathbb{C}^d \otimes \mathbb{C}^d$ (abbreviated $d \otimes d$) and its associated projector are defined by
\begin{equation}
\label{eq:maximally-entangled-state}
\ket{\Omega_d} = \frac{1}{\sqrt{d}} \sum_{j \in \mathcal{J}} \ket{j}_A \ket{j}_B, \qquad P_0 = \ketbra{\Omega_d}{\Omega_d}.
\end{equation}
Let $\{U_x : x \in \Lambda\}$ be a unitary error basis on $\mathbb{C}^d$, with index set $\Lambda$ of cardinality $\vert{}\Lambda\vert{} = d^2$, satisfying the orthogonality condition $\Tr(U_x^\dagger U_y) = d  \delta_{xy}$. Unitary error bases are in one-to-one correspondence with orthonormal maximally entangled
bases~\cite{Werner2001}. The associated Bell states and projectors are defined by
\begin{equation}
\label{eq:general-bell}
\ket{\Omega_x} := (I_d \otimes U_x) \ket{\Omega_d}, \qquad P_x := \ketbra{\Omega_x}{\Omega_x}.
\end{equation}
By the orthogonality condition, \(\{P_x\}_{x\in\Lambda}\) is a decomposition of the identity into mutually orthogonal rank-one projectors. For any probability distribution $c=\{c_x\in\mathbb{R}^{d^2}:x\in\Lambda\}$ satisfying $c_x\geq 0$ and $\sum_{x\in\Lambda}c_x=1$, the associated Bell-diagonal state is defined as
\begin{equation}\label{eq:general-BD}
 \rho(c)=\sum_{x\in\Lambda}c_xP_x,
\end{equation}
where the set of probabilities $c$ constitutes the spectrum of $\rho(c)$. We denote the set of all $d \otimes d$ Bell-diagonal states by $\BD(d)$. For the Weyl and Pauli systems (defined in the next subsection), where the unitary operators satisfy \(U_xU_y=e^{i\phi(x,y)}U_yU_x\), the correlated twirl operation is defined by 
\begin{equation}
\label{eq:general-twirl}
 \mathcal T(A)=\frac1{d^2}\sum_{x\in\Lambda}
 (U_x^*\otimes U_x)A(U_x^T\otimes U_x^\dagger).
\end{equation}
Alternatively, in these systems, the twirling channel can be expressed in the corresponding Bell basis as
\begin{equation}
\label{eq:twirl-pinching}
 \mathcal T(A)=\sum_x\Tr(P_xA)P_x,
\end{equation}
which is the Bell-basis projection associated with the unitary error basis~\cite{Werner2001}. It is a shared-randomness mixture of local unitaries and therefore preserves separability. In particular, the Bell coefficients $\{c_x\}_{x\in\Lambda}$ of a twirled product state $\mathcal T(\ketbra{a,b}{a,b})$ are the probabilities
\(c_x=|\braket{\Omega_x}{a,b}|^2\).

\subsection{Two particular unitary error bases}
\label{subsec:errorbases}
\textit{\textbf{Two-qutrit Bell states.}} To define two-qutrit Bell states, we use the notion of Heisenberg-Weyl unitary error basis. Let \(\omega=e^{2\pi i/d}\) and
\begin{equation}\label{eq:qutrit-Weyl}
 X\ket{j}=\ket{ j+1\pmod d},
 \qquad Z\ket{ j}=\omega^j\ket{ j}.
\end{equation}
Let $V_d := \mathbb{F}_d^2$ denote the discrete phase space. For any element $x = (k,l) \in V_d$, the associated Weyl–Heisenberg displacement operator $U_x = X^l Z^k$ yields the Bell basis state
\begin{equation}
\label{eq:qutrit-Bell}
\ket{\Omega_{x=(k,l)}} = (I_d \otimes U_x) \ket{\Omega_d} = \frac{1}{\sqrt{d}} \sum_{j=0}^{d-1} \omega^{kj} \ket{j, j+l}.
\end{equation}

$V_{d}$ is a symplectic vector space equipped with a nondegenerate antisymmetric bilinear symplectic product $[.,.]:V_{d}\times V_{d}\to\F_d$. For two elements $x,y\in\F_d^{2}$ with $y=(k',l')$, the symplectic product is defined as
\begin{equation}
\label{eq:qutrit-symp}
 [(k,l),(k',l')]=kl'-lk'\pmod d,
\end{equation}
which encodes the Weyl commutation relation
\begin{equation}
    \label{eq:weyl_comm-3x3}
    U_xU_y=\omega^{[x,y]}U_yU_x.
\end{equation}
A Lagrangian subspace $\mathcal{L}_{d}$ of a symplectic vector space $V_{d}$ is the maximal subspace satisfying $[x,y]=0~\forall~x,y\in \mathcal{L}_{d}\cap V_{d}$. Every one-dimensional subspace of \(V_3\) is Lagrangian and contains three points. There are four such $\mathcal{L}_{3}$ subspaces and three affine cosets of each, therefore, twelve affine Lagrangian phase-space lines.\\

\textit{\textbf{Two-ququart Bell states.}} For $4 \otimes 4$ Bell states, we adopt the stabilizer lattice-basis formalism by identifying each local ququart with two qubits~\cite{Dehaene2003, Hostens2005, BenattiLiguori2007, BenattiKarbalaii2013}. More generally, a $2^n \otimes 2^n$ Bell state, in this formalism, can be considered as a collection of $n$ distributed two-qubit Bell pairs. Specifically, the $i$-th pair is shared between two-dimensional local subsystems $A_i$ and $B_i$, such that the composite $2^n \otimes 2^n$ bipartite Bell state resides across the composite parties $A = A_1 A_2 \cdots A_n$ and $B = B_1 B_2 \cdots B_n$. Within this binary phase-space representation, we define the set and an element of it as
\begin{equation}
\label{eq:binary-labels}
V_{2n} := \mathbb{F}_2^{2n},\quad x = (p_1, q_1, \dots, p_n, q_n)\in V_{2n}.
\end{equation}
Utilizing the Hermitian single-qubit Pauli operators $Q_{p, q} = i^{pq} X^p Z^q$ with $(p, q) \in \mathbb{F}_2^2$, we construct the $n$-qubit Pauli strings
\begin{equation}
\label{eq:binary-pauli}
Q_x = \bigotimes_{i=1}^n Q_{p_i, q_i},
\end{equation}
which generate the orthonormal Bell basis
\begin{align}
\label{eq:ququart_Bell}
\ket{\Omega_{x=(p_i,q_i)^{\otimes n}}}_{AB} &= \left(I_{2^n}^A \otimes Q_x^B\right) \ket{\Omega_{2^n}}_{AB}  \\ &= S \left( \bigotimes_{i=1}^n (I_2^{A_i} \otimes Q_{p_i, q_i}^{B_i}) \ket{\Phi^+}_{A_i B_i} \right).\nonumber
\end{align}
Here, $\ket{\Omega_{2^n}} = \frac{1}{\sqrt{2^n}} \sum_{j \in \mathbb{F}_2^n} \ket{j}_A \ket{j}_B$ denotes the reference maximally entangled state, with $j = (j_1, \dots, j_n) \in \mathbb{F}_2^n$, and $\ket{\Phi^+} = \frac{1}{\sqrt{2}}(\ket{00} + \ket{11})$. The unitary operator $S$ reorders the subsystem tensor factors from $(A_1 B_1 \cdots A_n B_n)$ to $(A_1 \cdots A_n B_1 \cdots B_n)$.

In this setting, $V_{2n}$ is a symplectic vector space equipped with symplectic product $[.,.]:V_{2n}\times V_{2n}\to\F_2$. Given two elements $x,y\in\F_2^{2n}$ in phase-space coordinates with $y=(p_1',q_1',...,p_n',q_n')$, the symplectic product is defined as
\begin{align}
 [x,y]&=\sum_{i=1}^n p_iq_i'+q_ip_i'\pmod2,
 \label{eq:binary-symp}
\end{align}
which satisfies
\begin{equation}
\label{eq:q-refinement}
 Q_xQ_y=(-1)^{[x,y]}Q_yQ_x.
\end{equation}
Here the addition is bitwise XOR\footnote{Here and throughout, addition in \(\mathbb F_2^{2n}\) means bitwise exclusive OR (XOR), namely, componentwise addition modulo \(2\). Thus, for \(x,y\in\mathbb F_2^{2n}\), \( (x+y)_i=x_i+y_i\pmod 2, \)
so that \(0+0=1+1=0\) and \(0+1=1+0=1\).}. A Lagrangian subspace of \(V_4\) has dimension two, therefore $\mathcal{L}_4$ has four elements. There are fifteen such subspaces, each admitting four affine cosets, which yields a total of sixty affine Lagrangian planes. 
\subsection{Lagrangian-coset separability}
\label{subsec:lagra-coset-separability}
A useful consequence of the phase-space formulation is that it identifies a natural family of separable Bell-diagonal states. In both the two-qutrit Weyl system and the two-ququart binary-lattice system, a Lagrangian subspace labels a commuting family of operators, and the uniform mixture of Bell projectors supported on any affine coset of this subspace is separable. This common property is stated in the following lemma.
\begin{lemma}
\label{lemma:Coset_separability}
[Lagrangian-coset separability] Let \(\mathcal{L}\) be a Lagrangian subspace for either the qutrit Weyl--Heisenberg system or the binary two-qubit Pauli system. Thus, \(|\mathcal L|=d,\) where \(d=3\) for the qutrit system and \(d=4\) for the binary two-qubit system. For every affine coset \(L=t+\mathcal L\), the Bell-diagonal state
\[
\sigma_L=\frac{1}{d}\sum_{x\in L}P_x
\]
is separable.
\end{lemma}
We provide the proof of Lemma~\ref{lemma:Coset_separability} in Appendix~\ref{app:lemma1}. For $d=3$, the lemma gives the twelve affine-line states $\sigma_L=\frac{1}{3}\sum_{x\in L}P_x$, whose convex hull, called the kernel polytope of the two-qutrit magic simplex, consists entirely of separable states~\cite{Baumgartner2006,Baumgartner2007,Popp2022,Sutter2026}. For the \(4\otimes4\) lattice, it gives the sixty separable affine Lagrangian-plane states \( \sigma_L=\frac{1}{4}\sum_{x\in L}P_x.\) In the lattice-state literature, these four-point affine Lagrangian planes are described as special quadruples generated
by commuting Pauli operators~\cite{BenattiKarbalaii2013}. 
\subsection{PPT, absolutely PPT, and PermPPT}

Without loss of generality, let $\Gamma_A$ denote the partial transposition with respect to the first subsystem $A$. A bipartite state $\rho$ on $\mathbb{C}^{d_A} \otimes \mathbb{C}^{d_B}$ has a positive partial transpose (PPT) if $\rho^{\Gamma_A} \geq 0$; otherwise, it is non-positive partial transpose (NPT). The state is called absolutely PPT ($\APPT$) if $(W\rho W^\dagger)^{\Gamma_A} \geq 0$ for all global unitaries $W \in \mathrm{U}(d_A d_B)$, where $\mathrm{U}(d)$ denotes the set of all $d$-dimensional unitary matrices. Similarly, $\rho$ is absolutely separable (AS) if $W\rho W^\dagger$ remains separable under any global unitary $W$. Evidently, both $\APPT$ and AS are purely spectral properties of the density matrix $\rho$. 

Hildebrand~\cite{Hildebrand2007} gave a necessary and sufficient spectral criterion for the $\APPT$ property in arbitrary dimension. For example, in qutrit-qudit $(3\otimes d)$ systems with $d\ge3$, there are two linear matrix inequalities (LMIs) on the eigenvalues of a state that are together necessary and sufficient for the state to be $\APPT$. Specifically, let \( \{\lambda=(\lambda_1,...,\lambda_{3d}):\lambda_1\geq\lambda_2\geq\cdots\geq\lambda_{3d}\geq0,~\sum_i\lambda_i=1\} \) denote the ordered eigenvalues of $\rho$. Then $\rho$ is $\APPT$ if and only if it satisfies both of the following LMIs
\begin{equation}
\label{eq:APPT-3xd}
M_1(\lambda)\succeq0, \qquad M_2(\lambda)\succeq0,
\end{equation}
where
\begin{equation}
\label{eq:APPT_1}
M_1(\lambda)=
\begin{pmatrix}
2\lambda_{3d} & \lambda_{3d-1}-\lambda_1 & \lambda_{3d-3}-\lambda_2 \\
\lambda_{3d-1}-\lambda_1 & 2\lambda_{3d-2} & \lambda_{3d-4}-\lambda_3 \\
\lambda_{3d-3}-\lambda_2 & \lambda_{3d-4}-\lambda_3 & 2\lambda_{3d-5}
\end{pmatrix}
\end{equation}
and
\begin{equation}
\label{eq:APPT_2}
M_2(\lambda)=
\begin{pmatrix}
2\lambda_{3d} & \lambda_{3d-1}-\lambda_1 & \lambda_{3d-2}-\lambda_2 \\
\lambda_{3d-1}-\lambda_1 & 2\lambda_{3d-3} & \lambda_{3d-4}-\lambda_3 \\
\lambda_{3d-2}-\lambda_2 & \lambda_{3d-4}-\lambda_3 & 2\lambda_{3d-5}
\end{pmatrix}.
\end{equation}
However, this approach becomes infeasible in higher dimensions, because the number of LMIs grows exponentially with $\min(d_1,d_2)$ for $d_1\otimes d_2$ dimension.  For instance, Ref.~\cite{XiongSze2026} showed that the \(4\otimes d\) (\(d\ge4\)) $\APPT$ criterion can be expressed through twelve LMIs. Moreover, although these LMIs completely characterize $\APPT$ states in any bipartite dimension, the resulting constraints are nonlinear in the eigenvalues. For example, imposing positivity on the matrices in Eq.~\eqref{eq:APPT-3xd} yields several nonlinear constraints on $\lambda$. As a consequence, the set of $\APPT$ spectra (hereafter, we denote by $\APPT_{d_1,d_2}$ the set of all $d_1\otimes d_2$ APPT states; since APPT is a spectral property, we also use the same symbol for the set of their spectra, and write simply $\APPT$ when the dimensions are clear from context) forms a spectrahedron with curved boundaries, thereby possessing infinitely many extreme points. 

Recently, in Ref.~\cite{AhiableKothakondaWinter2026}, the authors provided a comprehensive characterization of the faces and extreme points of the $\APPT_{d_1,d_2}$ set for arbitrary local dimensions $(d_1, d_2)$. Nevertheless, relaxing these constraints yields a linearized condition that remains necessary, though no longer sufficient. In particular, for a $d_1\otimes d_2$ dimensional bipartite system, the following spectral condition is necessary for $\APPT$~\cite{Hildebrand2007} (see Ref.~\cite{Verstraete2001} for two-qubit case):
\begin{equation}
\lambda_1\le\lambda_{d_1d_2-1}+2\sqrt{\lambda_{d_1d_2}\lambda_{d_1d_2-2}},
\end{equation}
where $\lambda_1\ge\lambda_2\ge\cdots\ge\lambda_{d_1d_2}\ge0$ are the ordered eigenvalues of the state $\rho$ and $\sum_i\lambda_i=1$. Combining this with the AM-GM inequality, $2\sqrt{xy}\le x+y$ for $x,y\ge0$, we obtain the following lemma.
\begin{lemma}
\label{lemma:linearized_necessary_APPT}
[Linearized necessary condition for $\APPT$] Every $d_1\otimes d_2$ $\APPT$ state satisfies
\begin{equation}
\label{eq:APPT_Necessary}
\lambda_1 \leq \lambda_{d_1d_2} + \lambda_{d_1d_2-1} + \lambda_{d_1d_2-2}.
\end{equation}
\end{lemma}

We now introduce some notation used throughout the following sections.

\textit{\textbf{Ordered simplex.}} The set of ordered eigenspectra $\Delta_{d_1d_2}^{\downarrow} = \{ \lambda\in\mathbb{R}^{d_1d_2} : \lambda_1 \ge \cdots \ge \lambda_{d_1d_2} \ge 0,\, \sum_{i=1}^{d_1d_2} \lambda_i = 1\}$ forms a $(d_1d_2-1)$-simplex in $\mathbb{R}^{d_1d_2}$, whose extreme points are given by $v_k=\frac{1}{k}(1^{\times k},0^{\times (d_1d_2-k)})$ for $1\leq k \leq d_1d_2$. Consequently, $\Delta_{d_1d_2}^{\downarrow} = \operatorname{conv}(v_1, \dots, v_{d_1d_2})$.\\

\textit{\textbf{APPT outer polytope.}} The necessary condition in Eq.~\eqref{eq:APPT_Necessary} defines a linear half-space that, intersected with the ordered simplex $\Delta_{d_1d_2}^{\downarrow}$, forms an outer polytope enclosing the ordered $\APPT_{d_1,d_2}$ set. Letting $h_{d_1d_2}(\lambda) = \lambda_{d_1d_2} + \lambda_{d_1d_2-1} + \lambda_{d_1d_2-2} - \lambda_1$, this outer polytope is defined as $Q_{d_1d_2}^{\downarrow} = \Delta_{d_1d_2}^\downarrow \cap \{ \lambda : h_{d_1d_2}(\lambda) \ge 0 \}$.\\

\noindent\textit{\textbf{Permutation-PPT.}} Although $\APPT$ is a purely spectral property, studying whether every $\APPT$ state diagonal in a fixed basis is separable provides valuable insight into the $\APPT = \mathrm{AS}$ conjecture. Indeed, the equivalence holds if and only if all $\APPT$ states diagonal in every choice of orthonormal basis are separable; conversely, identifying a single basis containing an entangled $\APPT$ state would immediately refute the conjecture. To this end, we fix a Bell basis and analyze the separability of $\APPT$ states diagonal in that basis. The Bell-basis permutation unitaries $U_\pi$, $\pi\in S_{d^2}$, defined by $U_\pi P_xU_\pi^\dagger=P_{\pi(x)}$, are among the global unitaries that map the set of $d\otimes d$ Bell-diagonal states to itself. Their action on the coefficient vector is given by
\begin{equation}
\label{eq:perm-action}
(\pi c)_{\pi(x)} = c_x, \qquad \rho(\pi c) = \sum_x (\pi c)_x P_x,
\end{equation}
which corresponds to mapping the coefficient originally at index $x$ to the new position $\pi(x)$. We say that $c$ is \(\PPPT\) if $\rho(\pi c)$ is PPT for all permutations $\pi \in S_{d^2}$, and we denote the set of all such states by $\PPPT$. Since APPT requires PPT preservation under every global unitary, it immediately follows that
\begin{equation}\label{eq:APPT-implies-PPPT}
\APPT \cap \BD \subset \PPPT.
\end{equation}
The reverse inclusion, however, is not assumed to hold in general.

\section{Separability of two-qutrit APPT Bell-diagonal states}
\label{sec:3x3results}
Consider a two-qutrit state $\rho$ with ordered spectrum $\lambda_1 \ge \lambda_2 \ge \cdots \ge \lambda_9 \ge 0$, satisfying $\sum_{i=1}^{9} \lambda_i = 1$. Specializing the LMIs in Eqs.~\eqref{eq:APPT-3xd}–\eqref{eq:APPT_2} to \(d=3\), we define
\[ M_i^{3\otimes3}(\lambda) := \left.M_i(\lambda)\right|_{d=3}, \qquad i=1,2. \]
The two-qutrit $\APPT$ criterion therefore becomes
\[
\rho\in\APPT_{3,3} \quad\Longleftrightarrow\quad M_1^{3\otimes3}(\lambda)\succeq0
\ \text{and}\ M_2^{3\otimes3}(\lambda)\succeq0.
\]
\subsection{Ordered simplex, ordered APPT set, and outer polytope}
As introduced earlier, the ordered simplex in the two-qutrit setting is defined by $\Delta_9^\downarrow = \{ \lambda \in \mathbb{R}^9 : \lambda_1 \ge \lambda_2 \ge \cdots \ge \lambda_9 \ge 0,\, \sum_{i=1}^9 \lambda_i = 1 \}$. Its nine extreme points are given by
\begin{equation}
\label{eq:vertices_ordered_simplex}
u_k = \frac{1}{k} (1^{\times k}, 0^{\times (9-k)}), \quad k \in {1, \dots, 9},
\end{equation}
so that $\Delta_9^\downarrow = \operatorname{conv}(u_1, \dots, u_9)$. Setting $\lambda_{10} := 0$, every ordered spectrum admits the unique barycentric decomposition
\begin{equation}
\label{eq:barycentric_3x3}
\lambda = \sum_{k=1}^9 \alpha_k u_k, \quad \alpha_k = k(\lambda_k - \lambda_{k+1}) \ge 0, \quad \sum_{k=1}^9 \alpha_k = 1.
\end{equation}
By Lemma~\ref{lemma:linearized_necessary_APPT}, any two-qutrit $\APPT$ spectrum must satisfy the linear inequality $\lambda_1 \le \lambda_7 + \lambda_8 + \lambda_9$. Defining the linear functional $h_9(\lambda) = \lambda_7 + \lambda_8 + \lambda_9 - \lambda_1$, we introduce the $3 \otimes 3$ $\APPT$ outer polytope 
\begin{equation}
    \label{eq:APPT_outer_polytope_3x3}
    Q_9^\downarrow = \Delta_9^\downarrow \cap \{ \lambda : h_9(\lambda) \ge 0 \}.
\end{equation} 
In what follows, we show that the set $Q_9^\downarrow$, together with all distinct permutations of its elements, corresponds to separable Bell-diagonal states.

To this end, we first determine the extreme points of the outer polytope $Q_9^\downarrow$. The following lemma provides a complete rank characterization of its vertices in barycentric coordinates.
\begin{lemma}
    \label{lemma:rank_vertices_outerpolytope}
    [Rank characterization of vertices]
    Let $\lambda = \sum_{k=1}^9 \alpha_k u_k \in Q_9^\downarrow$ with $\alpha_k \ge 0$ and $\sum_{k=1}^9 \alpha_k = 1$. Defining the barycentric rank of $\lambda$, denoted by $\operatorname{rank}(\lambda)$, as the number of non-zero coefficients in its barycentric decomposition, the following conditions hold:
    \begin{enumerate}
    \item If $h_9(\lambda) > 0$ and $\operatorname{rank}(\lambda) \ge 2$, then $\lambda$ is not a vertex of $Q_9^\downarrow$.
    \item If $h_9(\lambda) = 0$ and $\operatorname{rank}(\lambda) \ge 3$, then $\lambda$ is not a vertex of $Q_9^\downarrow$.
    \end{enumerate}
\end{lemma}
For the proof of the lemma, see Appendix~\ref{app:lemma1}. Thus every new vertex produced by the cut must lie on an original edge of $\Delta_9^\downarrow$.
\subsubsection{The fifteen ordered vertices}
Using the vertex structure as given in Eq.~\eqref{eq:vertices_ordered_simplex}, a direct evaluation yields
\begin{equation*}
h_9(u_k) = -\frac{1}{k} < 0 \quad \text{for } k \in {1, \dots, 6},~\text{and}
\end{equation*} 
\begin{equation*}
h_9(u_7) = 0, \qquad h_9(u_8) = \frac{1}{8} > 0, \qquad h_9(u_9) = \frac{2}{9} > 0.
\end{equation*}
Thus, the original simplex vertices retained in the outer polytope are $u_7$, $u_8$, and $u_9$. Each of the six infeasible vertices $u_i$ ($1 \le i \le 6$) forms a crossing edge with both $u_8$ and $u_9$, generating twelve new extreme points on the supporting hyperplane $h_9(\lambda) = 0$. Specifically, for $i \in \{1, \dots, 6\}$, the intersection of the line segment $(u_i, u_8)$ with the plane $h_9(\lambda) = 0$ is
\begin{equation}
\label{eq:wi8}
w_i^{(8)} = \left(\left(\frac{2}{i+8}\right)^{\times i}, \left(\frac{1}{i+8}\right)^{\times (8-i)}, 0\right),
\end{equation}
while the intersection of $(u_i, u_9)$ with $h_9(\lambda) = 0$ is
\begin{equation}
\label{eq:wi9}w_i^{(9)} = \left(\left(\frac{3}{2i+9}\right)^{\times i}, \left(\frac{1}{2i+9}\right)^{\times (9-i)}\right).
\end{equation}
Consequently, the outer polytope is given by the convex hull of these fifteen extreme points:
\begin{equation}
\label{eq:Q15}
Q_9^\downarrow = \operatorname{conv}\bigl({u_7, u_8, u_9}, {w_i^{(8)}, w_i^{(9)} : i = 1, \dots, 6}\bigr).
\end{equation}
The explicit coordinate forms of all fifteen ordered vertices are tabulated in Appendix~\ref{app:ordered_vertices}. A corresponding characterization of the $\APPT$ outer polytope for arbitrary dimensions was also recently reported in Ref.~\cite{Tran2026}.
\subsection{From ordered simplex to full probability simplex}
Thus far, our characterization of the $\APPT$ set and its outer polyhedral relaxation has been restricted to the ordered simplex, where each point represents a descending eigenspectrum. We now extend this construction to the full probability simplex, $\Delta_9 = \{ c \in \mathbb{R}^9 : c_i \ge 0,\, \sum_{i=1}^9 c_i = 1 \}$, where $c$ represents the eigenspectrum (or, the coefficient vector) associated with the spectral decomposition of an arbitrary two-qutrit state. Toward this end, let the symmetric group $S_9$ act on $\mathbb{R}^9$ by coordinate permutations. Denoting the fifteen vertices of $Q_9^\downarrow$ by $\{V_r\}_{r=1}^{15}$, we define the set
\begin{equation}
\mathcal{V}_{2725} = \bigcup_{r=1}^{15} \{ P_\pi V_r : \pi \in S_9 \},
\label{eq:V2725}
\end{equation}
where duplicate vectors are eliminated, and we define its convex hull as
\begin{equation}
\mathcal{C}_{2725} = \operatorname{conv}(\mathcal{V}_{2725}).
\label{eq:C2725}
\end{equation}
Here, $P_\pi$ represents the permutation matrix corresponding to $\pi \in S_9$, acting on $c \in \Delta_9$ via $(\pi c)_{\pi(x)} = c_x$, which places the entry originally at coordinate $x$ into position $\pi(x)$. If the distinct entries in a vertex $V_r$ have multiplicities $m_1, m_2, \dots, m_k$ (such that $\sum_j m_j = 9$), the size of its permutation orbit is given by the multinomial coefficient
\begin{equation}
N_r = \frac{9!}{m_1! , m_2! \cdots m_k!}.
\end{equation}
For instance, the vertex $w_6^{(9)} = \bigl((1/7)^{\times 6}, (1/21)^{\times 3}\bigr)$ has multiplicities $m_1 = 6$ and $m_2 = 3$, yielding $9! / (6! \, 3!) = 84$ distinct permutations. Evaluating these orbit sizes across all fifteen ordered vertices yields
\begin{eqnarray}
\{N_r\}_{r=1}^{15} &=&\{1, 9, 36, 84, 252, 126, 504, 126, 630, \nonumber\\ &&84, 504, 36, 252, 9, 72\},
\end{eqnarray}
whose sum is precisely $\sum_{r=1}^{15} N_r = 2725$.

\begin{proposition}
\label{prop:APPT_inside_outerpolytope}
Every two-qutrit $\APPT$ spectrum belongs to $\mathcal C_{2725}$.
\end{proposition}
\begin{proof}
Let $c \in \Delta_9$ be the eigenspectrum of a two-qutrit $\APPT$ state, and let $\lambda = c^\downarrow$ denote its decreasing rearrangement. Because $\APPT$ is a spectral property, it follows directly that $\lambda \in Q_9^\downarrow$. Relabeling the fifteen vertices of $Q_9^\downarrow$ from Eq.~\eqref{eq:Q15} as $\{V_r\}_{r=1}^{15}$, the ordered spectrum admits the convex decomposition
\begin{equation*}
\lambda = \sum_{r=1}^{15} t_r V_r, \quad \text{with } t_r \ge 0 \quad \text{and} \quad \sum_{r=1}^{15} t_r = 1.
\end{equation*}
There exists a permutation $\pi \in S_9$ such that $c = P_\pi \lambda$. By linearity of the permutation matrix $P_\pi$,
\begin{equation*}
c = \sum_{r=1}^{15} t_r P_\pi V_r.
\end{equation*}
Since each $P_\pi V_r \in \mathcal{V}_{2725}$ by the definition in Eq.~\eqref{eq:V2725}, it follows immediately that $c \in \operatorname{conv}(\mathcal{V}_{2725}) = \mathcal{C}_{2725}$.
\end{proof}
\subsection{Characterization of the Bell-diagonal APPT outer polytope}
Let the nine two-qutrit Bell projectors $P_x$, defined in Eq.~\eqref{eq:qutrit-Bell}, be labeled by phase-space points $x = (k, l) \in \mathbb{F}_3^2$. The two-qutrit Bell-diagonal state associated with an arbitrary probability vector $c \in \Delta_9$ is then given by
\begin{equation}
\label{eq:two-qutrit_BDS}
\rho(c) = \sum_{x \in \mathbb{F}_3^2} c_x P_x.
\end{equation}
In what follows, we characterize the outer polytope enclosing the set of all $\APPT$ Bell-diagonal states by decomposing the distinct permutations of its ordered vertices into orbits under the affine symmetry.

\subsubsection*{Orbit decomposition of the fifteen ordered spectral types}
In order to determine the orbit decomposition of the fifteen ordered spectral types under the action of $\mathrm{AGL}(2,3)$, we analyze the incidence and collinearity properties of each spectral profile when mapped onto the affine plane $\mathrm{AG}(2,3)$. The action of the affine symmetry group on the two-qutrit Weyl--Heisenberg Bell simplex is analyzed in detail in Appendix~\ref{app:3x3_symmetry}. For compactness, we denote the distinct entries within a given spectral type by $H$, $M$, $L$, and $Z$, standing respectively for high, intermediate, low, and zero coefficients (when present). A configuration is uniquely specified by the phase-space coordinates supporting the non-generic values; the remaining positions are then occupied automatically. For instance, for the spectrum $((2/13)^{\times 5}, (1/13)^{\times 3}, 0)$, fixing the locations of the three $1/13$ entries and the single $0$ entry completely determines the five positions carrying $2/13$. In Appendix~\ref{app:rep_count_3x3}, we detail the representative counting arguments for all non-trivial spectral types where orbit splitting is not immediate. Complementing this analysis, Appendix~\ref{app:numerical_algorithm} provides an algorithm to systematically extract both the orbit sizes and canonical representatives numerically. The numbers of affine orbits generated by each of the fifteen spectral types are
\begin{equation*}
(1, 1, 1, 2, 2, 2, 3, 2, 4, 2, 3, 1, 2, 1, 1),
\end{equation*}
yielding exactly $28$ distinct orbits in total. The corresponding orbit sizes, grouped according to the fifteen ordered spectral types, are given by
\begin{equation*}
\begin{split}
&\{1\},\{9\},\{36\},\{72,12\},\{216,36\},\{72,54\},
\{216,216,72\},\\
&\{72,54\},\{216,216,144,54\},\{72,12\},
\{216,216,72\},\{36\},\\&\{216,36\},\{9\},\{72\},
\end{split}
\end{equation*}
whose sum recovers the total count of $2725$. Table~\ref{tab:orbits} summarizes the detailed orbit decomposition of these $2725$ coefficient vectors. Specifically, the third column of the table reports the number of all distinct $S_9$ permutations of the specific ordered spectral type, while the last column gives the sizes of the inequivalent $\AGL(2,3)$ orbits within that spectral type. In each row, the affine orbit sizes sum to the corresponding $S_9$ multiplicity for that type, and summing over all rows recovers $\vert{}\mathcal{V}_{2725}\vert{} = 2725$. Appendix~\ref{app:orbit_representative_3x3} lists all $28$ orbit representatives. We emphasize that these choices are not unique: any element $c' \in \operatorname{Orb}(c)$ serves as an equivalent representative for the orbit.

\begin{table*}[t]
\caption{Decomposition of the 2725 coefficient vectors into distinct affine orbits. The third column gives the number of all distinct $S_9$ multiplicities of the coefficient multiset. The last column gives the sizes of the inequivalent $\AGL(2,3)$ orbits within that spectral type. Here $H,M,L,Z$ denote, within a given spectral type, the distinct coefficient values, standing respectively for higher, intermediate, lower, and zero coefficients (if exist); these labels are row-dependent. A ``marked set'' refers to the subset of phase-space points whose coefficient values are specified explicitly, the remaining positions then being fixed automatically. A ``four-cap'' is a four-point subset of $\text{AG}(2,3)$ containing no three collinear points.}
\label{tab:orbits}
\centering
\small
\renewcommand{\arraystretch}{1.35}
\setlength{\tabcolsep}{10pt}
\begin{tabular}{c l r p{6.5cm} l}
\toprule
Type & Ordered coefficient multiset & $S_9$ count & Affine-geometric distinction & Orbit sizes \\
\midrule
1  & $(1/9)^{\times9}$ & 1 & uniform & $1$ \\
2  & $(1/8)^{\times8},0$ & 9 & one distinguished point & $9$ \\
3  & $(1/7)^{\times7},0^{\times2}$ & 36 & unordered pair of zero positions; their line has pattern $HZZ$ & $36$ \\
4  & $(1/7)^{\times6},(1/21)^{\times3}$ & 84 & three $L$ points collinear / noncollinear & $12,72$ \\
5  & $(1/7)^{\times6},(1/14)^{\times2},0$ & 252 & $ZMM$ collinear / noncollinear & $36,216$ \\
6  & $(3/19)^{\times5},(1/19)^{\times4}$ & 126 & four $L$ points contain a line / form a four-cap & $72,54$ \\
7  & $(2/13)^{\times5},(1/13)^{\times3},0$ & 504 & $L^3$ collinear; or $L^3$ noncollinear with $Z$ on a secant; or four-cap & $72,216,216$ \\
8  & $(3/17)^{\times4},(1/17)^{\times5}$ & 126 & equivalently classify the four $H$ points: line-containing / four-cap & $72,54$ \\
9  & $(1/6)^{\times4},(1/12)^{\times4},0$ & 630 & for marked set $H^4Z$: $HHH$ only; $ZHH$ only; one of each; two $ZHH$ lines & $144,216,216,54$ \\
10 & $(1/5)^{\times3},(1/15)^{\times6}$ & 84 & three $H$ points collinear / noncollinear & $12,72$ \\
11 & $(2/11)^{\times3},(1/11)^{\times5},0$ & 504 & analogous to type 7 after exchanging the distinguished coefficient class & $72,216,216$ \\
12 & $(3/13)^{\times2},(1/13)^{\times7}$ & 36 & unordered pair of $H$ positions & $36$ \\
13 & $(1/5)^{\times2},(1/10)^{\times6},0$ & 252 & $ZHH$ collinear / noncollinear & $36,216$ \\
14 & $3/11,(1/11)^{\times8}$ & 9 & one distinguished point & $9$ \\
15 & $2/9,(1/9)^{\times7},0$ & 72 & ordered pair of distinct distinguished points $(H,Z)$ & $72$ \\
\bottomrule
\end{tabular}
\end{table*}
\subsection{Separability of the Bell-diagonal APPT outer polytope}
We now proceed to prove that every Bell-diagonal state $\rho(c)$, as defined in Eq.~\eqref{eq:two-qutrit_BDS}, associated with a coefficient vector $c \in \mathcal{C}_{2725}$ is separable. Because $\mathcal{C}_{2725}$ is the convex hull of the set of coefficient vectors $\mathcal{V}_{2725}$, establishing separability reduces by convexity to showing that $\rho(c)$ is separable for all $c \in \mathcal{V}_{2725}$. Since the affine group $\mathrm{AGL}(2,3)$ preserves the entanglement structures and partitions $\mathcal{V}_{2725}$ into 28 disjoint affine orbits, it suffices to certify the separability of a single target representative from each orbit. To this end, we construct an explicit inner separable polytope within the two-qutrit Bell-diagonal magic simplex.

\subsubsection{An explicitly separable Bell-diagonal generating family}
We construct a finite family of certified separable Bell-diagonal states that serves as a generating set to certify the separability of $2725$ target Bell-diagonal states.

\paragraph{Affine lines in $\mathbb{F}_3^2$.} Identifying the nine Bell-projector indices with the discrete phase space $\mathbb{F}_3^2$, there exist twelve affine lines. For each line $\ell \subset \mathbb{F}_3^2$, define the uniform mixture
\begin{equation}
\mathfrak{L}_\ell = \frac{1}{3} \sum_{x \in \ell} P_x.
\label{eq:line}
\end{equation}
By Lemma~\ref{lemma:Coset_separability}, each line state $\mathfrak{L}_\ell$ is separable~\cite{Baumgartner2006,Baumgartner2007}.

\paragraph{A non-trivial separable state and its affine orbit.} Consider the Bell coefficient vector
\begin{equation}
\mathfrak{R}_0 = \left(\frac{2}{9}, 0, \frac{2}{9}, \frac{4}{27}, \frac{4}{27}, \frac{4}{27}, \frac{1}{27}, \frac{1}{27}, \frac{1}{27}\right),
\label{eq:R0}
\end{equation}
ordered according to the phase-space coordinates $(k, l) \in \mathbb{F}_3^2$ as $(00, 01, 02, 10, 11, 12, 20, 21, 22)$. This state is separable. An explicit product-state realization is given by $\ket{a, b} = \ket{a} \otimes \ket{b}$ with
\[
\ket{a}=
\begin{pmatrix}
\sqrt{3+\sqrt6}/3\\
i\sqrt{3-\sqrt6}/3\\
\sqrt{(3-\sqrt6)/18}-i\sqrt{(3+\sqrt6)/18}
\end{pmatrix},
\]
\[
W=
\begin{pmatrix}
0&0&-(\sqrt3+i)/2\\
0&i&0\\
(\sqrt3-i)/2&0&0
\end{pmatrix},
\;\;\;
\ket{b}=W\overline{\ket{a}}.
\]
A straightforward calculation confirms that $\vert{}\langle\Omega_{k,l}\vert{}a,b\rangle\vert{}^2 = (\mathfrak{R}_0)_{k,l}$. Therefore, by Eq.~\eqref{eq:twirl-pinching}, the Bell-twirl $\mathcal{T}(\ket{a,b}\bra{a,b})$ yields the Bell-diagonal state $\rho(\mathfrak{R}_0)$. Because the Bell-twirling map is local operations and classical communication (LOCC), it preserves separability, guaranteeing that $\rho(\mathfrak{R}_0)$ is separable. As established,  the affine group operation on the coefficient vector preserves separability of associated Bell-diagonal states. Because the stabilizer $\operatorname{Stab}(\mathfrak{R}_0)$ contains exactly six affine transformations, the orbit-stabilizer relation as given in Eq.~\eqref{eq:orbitstabilizer} together with $\vert{}\mathrm{AGL}(2,3)\vert{} = 432$ implies that the affine orbit contains $\vert{}\operatorname{Orb}(\mathfrak{R}_0)\vert{} = 432 / 6 = 72$ distinct separable Bell-diagonal states. Combining these states yields a finite separable generating set of $84$ extreme states:
\begin{equation}
\mathcal{G}{\mathrm{sep}} = \{\mathfrak{L}_\ell\}_{\ell=1}^{12} \cup \operatorname{Orb}_{\AGL(2,3)}(\mathfrak{R}_0),
\label{eq:Gsep}
\end{equation}
every element of which is manifestly separable.

\subsubsection{Exact rational decomposition certificate}
Let the $84$ explicitly separable generators of $\mathcal G_{\rm sep}$ be denoted by $G_1,\ldots,G_{84}$. For each of the $28$ affine-orbit representatives $T_r$, $r=1,\ldots,28$, listed in Appendix~\ref{app:orbit_representative_3x3}, we provide an exact separability certificate. Specifically, for each $T_r$ there exists a finite support $S_r\subseteq\{1,\ldots,84\}$ and positive rational coefficients $p_{rj}$ such that
\begin{equation}
T_r=\sum_{j\in S_r}p_{rj}G_j,
\qquad
\sum_{j\in S_r}p_{rj}=1.
\label{eq:cert}
\end{equation}
The complete ordered list of the generators \(G_1,\ldots,G_{84}\), together with the exact \textit{Wolfram Mathematica} implementation used to obtain and verify the decompositions in Table~\ref{tab:84cert1}, is provided in the Supplementary Material~\cite{SupplementalMaterial}. Since every $G_j\in\mathcal G_{\rm sep}$ is separable, Eq.~\eqref{eq:cert} gives an explicit convex decomposition of each $T_r$ into separable states and therefore proves the separability of all $28$ representatives. The supports $S_r$ and the corresponding exact rational coefficients are reported in Table~\ref{tab:84cert1}. The $28$ representatives themselves are listed in Appendix~\ref{app:orbit_representative_3x3} in the flattening order $(00,01,02,10,11,12,20,21,22)$. Each decomposition in Eq.~\eqref{eq:cert} is an exact identity over the rationals.

\begin{table*}[t]
\caption{Exact rational separable-decomposition certificates for the $28$ affine-orbit representatives. For each $r$, the support $S_r\subseteq\{1,\ldots,84\}$ and the listed weights $p_{rj}$ satisfy $T_r=\sum_{j\in S_r}p_{rj}G_j$. The representatives $\{T_r\}_{r=1}^{28}$ are listed in Appendix~\ref{app:orbit_representative_3x3}, and the ordered generators $G_1,\ldots,G_{84}$ are provided in the Supplementary Material~\cite{SupplementalMaterial}.}
\label{tab:84cert1}
\centering
\small
\setlength{\tabcolsep}{5pt}
\renewcommand{\arraystretch}{1.4}
\begin{tabular}{c c c p{4.0cm} p{10.0cm}}
\toprule
$r$ & type & orbit & support $S_r$ & exact weights $p_{rj}$ \\
\midrule
1 & 1 & 1 & 1,2,3 & $1/3, 1/3, 1/3$ \\
2 & 2 & 9 & 2,3,5,6,8,9,11,12 & $1/8$ (8 times) \\
3 & 3 & 36 & 2,3,6,9,11 & $2/7, 2/7, 1/7, 1/7, 1/7$ \\
4 & 4 & 12 & 1,2,3 & $1/7, 3/7, 3/7$ \\
5 & 4 & 72 & 1,2,3,5,6,7,9 & $1/21, 1/7, 5/21, 2/21, 4/21, 2/21, 4/21$ \\
6 & 5 & 36 & 2,3,5,6,8,9,11,12 & $2/7, 2/7, 1/14, 1/14, 1/14, 1/14, 1/14, 1/14$ \\
7 & 5 & 216 & 2,3,5,6,9,11,12 & $1/7, 3/14, 1/7, 3/14, 1/7, 1/14, 1/14$ \\
8 & 6 & 72 & 2,3,7,9,10,12,20,72 & $17/133, 31/133, 4/133, 4/133, 12/133, 20/133, 36/133, 9/133$ \\
9 & 6 & 54 & 1,2,3,6,7,9,11,12 & $1/19, 1/19, 5/19, 4/19, 2/19, 2/19, 2/19, 2/19$ \\
10 & 7 & 216 & 2,3,5,6,9,12 & $3/13, 4/13, 2/13, 2/13, 1/13, 1/13$ \\
11 & 7 & 216 & 2,3,5,6,9,11,12 & $1/13, 3/13, 1/13, 3/13, 1/13, 2/13, 2/13$ \\
12 & 7 & 72 & 2,3,5,6,9,11,12,65 & $4/39, 4/39, 4/39, 4/39, 2/13, 4/39, 4/39, 3/13$ \\
13 & 8 & 72 & 1,2,3,6,7,12 & $1/17, 3/17, 7/17, 2/17, 2/17, 2/17$ \\
14 & 8 & 54 & 2,3,5,6,7,10,20,48 & $2/17, 2/17, 1/17, 1/17, 1/17, 1/17, 9/34, 9/34$ \\
15 & 9 & 216 & 2,3,5,6,11,12 & $1/6, 1/3, 1/12, 1/6, 1/12, 1/6$ \\
16 & 9 & 54 & 2,3,5,6 & $1/4$ (4 times) \\
17 & 9 & 216 & 2,3,5,6,8,9,12,30 & $7/36, 7/36, 1/36, 1/9, 1/36, 1/9, 1/12, 1/4$ \\
18 & 9 & 144 & 3,5,6,8,9,11,12 & $1/4, 1/12, 1/6, 1/12, 1/6, 1/6, 1/12$ \\
19 & 10 & 12 & 1,2,3 & $1/5, 1/5, 3/5$ \\
20 & 10 & 72 & 2,3,6,7,8,20,38,48,84 & $2/95, 44/285, 2/15, 47/285, 1/19, 9/95, 3/19, 3/19, 6/95$ \\
21 & 11 & 72 & 2,3,5,6,8,9,11,12 & $1/11, 4/11, 1/11, 1/11, 1/11, 1/11, 1/11, 1/11$ \\
22 & 11 & 216 & 2,3,5,6,11,12,38 & $4/33, 7/33, 4/33, 7/33, 1/33, 1/33, 3/11$ \\
23 & 11 & 216 & 2,3,6,8,9,12 & $2/11, 3/11, 2/11, 1/11, 1/11, 2/11$ \\
24 & 12 & 36 & 3,5,6,7,8,10,11,75,84 & $3/13, 2/39, 2/39, 2/39, 2/39, 2/39, 2/39, 3/13, 3/13$ \\
25 & 13 & 216 & 2,3,5,6,8,11 & $1/10, 3/10, 1/5, 1/5, 1/10, 1/10$ \\
26 & 13 & 36 & 2,3,5,6,11,12,38,65 & $1/15, 1/15, 1/15, 1/15, 1/15, 1/15, 3/10, 3/10$ \\
27 & 14 & 9 & 1,2,3,6,8,10 & $1/11, 1/11, 3/11, 2/11, 2/11, 2/11$ \\
28 & 15 & 72 & 2,3,5,6,8,9 & $1/9, 2/9, 1/9, 2/9, 2/9, 1/9$ \\
\bottomrule
\end{tabular}
\end{table*}

\begin{proposition}
Every vector in $\mathcal V_{2725}$ corresponds to a separable Bell-diagonal state.
\label{prop:2725sep}
\end{proposition}

\begin{proof}
By Eq.~\eqref{eq:cert}, each affine-orbit representative $T_r$ can be expressed as a convex combination of states in $\mathcal G_{\rm sep}$ and is therefore separable. Now let $T$ be any of the $2725$ target vectors. By construction of the affine orbit decomposition, there exist an index $r\in\{1,\ldots,28\}$ and an element $g\in\operatorname{AGL}(2,3)$ such that
\[
T=g\cdot T_r.
\]
The action of $\AGL(2,3)$ on the Bell coefficients preserves separability. Consequently, since $T_r$ is separable, so is $g\cdot T_r=T$. The $28$ affine orbits form a disjoint partition of $\mathcal V_{2725}$ and together contain all $2725$ target vectors. Hence every element of $\mathcal V_{2725}$ corresponds to a separable Bell-diagonal state.
\end{proof}

The preceding results provide the two ingredients required for the final conclusion. Every absolutely PPT Bell-diagonal coefficient vector belongs to $\mathcal C_{2725}$, while every element of $\mathcal V_{2725}$ corresponds to a separable Bell-diagonal state. This leads directly to the following theorem.

\begin{theorem}
Every absolutely PPT two-qutrit Bell-diagonal state is separable. Equivalently,
\begin{equation}
\APPT_{3,3}\cap\BD(3)\subseteq\SEP .
\end{equation}
\end{theorem}

\begin{proof}
Let $\rho(c)$ be a two-qutrit Bell-diagonal state that is absolutely PPT. By Proposition~\ref{prop:APPT_inside_outerpolytope}, its coefficient vector $c$ belongs to the permutation-symmetric outer polytope,
\[
c\in\mathcal C_{2725}
=\operatorname{conv}(\mathcal V_{2725}).
\]
Hence there exist coefficients \(q_\nu\ge0\), with \(\sum_{\nu=1}^{2725}q_\nu=1\), such that
\[
c=\sum_{\nu=1}^{2725} q_\nu \bar{T}_\nu, \qquad \bar{T}_\nu\in\mathcal V_{2725}.
\]
For each \(\bar T_\nu\), let
\[
\rho({\bar T_\nu}) = \sum_{x\in\mathbb F_3^2}(\bar T_\nu)_x P_x
\]
be the corresponding Bell-diagonal state. By Proposition~\ref{prop:2725sep}, every \(\rho({\bar T_\nu})\) is separable.
Therefore, since
\[
\rho_c = \sum_{\nu=1}^{2725}q_\nu\,\rho({\bar T_\nu})
\]
is a convex mixture of separable states, it is also separable.
\end{proof}

\section{Separability of two-ququart APPT Bell-diagonal states}
\label{sec:4x4results}
As discussed earlier, the necessary and sufficient conditions for $\APPT$ in $4 \otimes 4$ systems are governed by twelve LMIs~\cite{XiongSze2026}. Nevertheless, by Lemma~\ref{lemma:linearized_necessary_APPT}, every $4 \otimes 4$ $\APPT$ spectrum satisfies the linear constraint
\begin{equation}
\lambda_1 \le \lambda_{14} + \lambda_{15} + \lambda_{16},
\label{eq:APPT_necessary_4x4}
\end{equation}
where $\lambda = (\lambda_1, \dots, \lambda_{16})$ denotes the sorted eigenspectrum of a $4 \otimes 4$ state $\rho$ with $\lambda_1 \ge \cdots \ge \lambda_{16} \ge 0$.

The ordered simplex is defined as $\Delta_{16}^{\downarrow} = \{ \lambda\in\mathbb{R}^{16} : \lambda_1 \ge \cdots \ge \lambda_{16} \ge 0,\, \sum_{i=1}^{16} \lambda_i = 1\}$, with extreme points are given by $v_k=\frac{1}{k}(1^{\times k},0^{\times (16-k)})$ for $1\leq k \leq 16$. The outer polytope is defined as $Q_{16}^{\downarrow} = \Delta_{16}^\downarrow \cap \{ \lambda : h_{16}(\lambda) \ge 0 \}$, where $h_{16}(\lambda) = \lambda_{16} + \lambda_{15} + \lambda_{14} - \lambda_1$. Extending Lemma~\ref{lemma:rank_vertices_outerpolytope} to $Q_{16}^\downarrow$ directly yields an explicit characterization of its 29 ordered vertices.  Specifically, the simplex vertices $v_{14}$, $v_{15}$ and $v_{16}$ are shared by $\Delta_{16}^{\downarrow}$ and $Q_{16}^{\downarrow}$, whereas each $v_i$ with $i=1,\dots,13$ lies outside $Q_{16}^{\downarrow}$, and its edges to both $v_{15}$ and $v_{16}$ cross the supporting hyperplane $h_{16}(\lambda)=0$. This produces $26$ new extreme points on that hyperplane. Therefore, the outer polytope is given by the convex hull of these vertices as
\begin{equation}
Q_{16}^{\downarrow}=\conv\left(v_{14},v_{15},v_{16},x_i^{(15)},x_i^{(16)}:i=1,2,...,13\right),
\end{equation}
where 
\begin{eqnarray}
    x_i^{(15)}&=&\left(\left(\frac{2}{15+i}\right)^{\times i},\left(\frac{1}{15+i}\right)^{\times (15-i)},0\right)\\
    x_i^{(16)}&=&\left(\left(\frac{3}{16+2i}\right)^{\times i},\left(\frac{1}{16+2i}\right)^{\times (16-i)}\right).
\end{eqnarray}
Labeling the sixteen two-ququart Bell projectors $P_x$, defined in Eq.~\ref{eq:ququart_Bell}, by phase-space coordinates $x=(p_1,q_1,p_2,q_2)\in\F_2^4)$, the Bell-diagonal state corresponding to a probability vector $c\in\mathbb{R}^{16}:c_i\ge0,~\sum_ic_i=1$ is then given by 
\begin{equation}
    \label{eq:eq:two-ququart_BDS}
    \rho(c)=\sum_{x\in\F_2^4}c_x p_x.
\end{equation}
In sharp contrast to the two-qutrit scenario, the ordered vertices of $Q_{16}^\downarrow$ do not all correspond to separable Bell-diagonal states. For instance, consider the vertex $x_3^{(16)} = \bigl((3/22)^{\times 3}, (1/22)^{\times 13}\bigr)$. Assigning its three dominant weights to the Bell projectors $P_{0000}$, $P_{0001}$, and $P_{0010}$ produces the Bell-diagonal state \( \tilde{\rho} = \frac{1}{22} \left(I_{16} + 2 P_{0000} + 2 P_{0001} + 2 P_{0010}\right).\) The partial transpose $\tilde{\rho}^{\Gamma_A}$ has a negative minimum eigenvalue of $-1/44$, certifying that $\tilde{\rho}$ is NPT entangled. 

This counterexample demonstrates that the similar technique which was used to prove the separability of two-qutrit Bell-diagonal $\APPT_{3,3}$ set cannot be directly extended to $4 \otimes 4$ systems. The underlying reason behind this failure is that the necessary condition in Eq.~\eqref{eq:APPT_necessary_4x4} is very weak. Numerically, we have found that, with respect to the flat spectral measure, $\APPT$ spectra occupy approximately \(2.22\%\) of the \(4\otimes4\) $\APPT$ outer polytope, compared with \(36.87\%\) in the two-qutrit case~\cite{APPTSamplingNote}. Consequently, even if all $\APPT$ Bell-diagonal states in $4 \otimes 4$ are separable, certifying this property requires going beyond the technique involving linearized outer polytope approximation. 

In what follows, we demonstrate that if a two-ququart Bell-diagonal state as given in Eq.~\eqref{eq:eq:two-ququart_BDS} is entangled, then there exists a permutation of its coefficient vector $c$ that yields an NPT state. Hence, an entangled $4\otimes4$ lattice Bell-diagonal state cannot be $\APPT$. Equivalently, every lattice Bell-diagonal state that is PPT under all permutations of its Bell coefficients is separable. Towards this aim, we first compute the eigenspectrum of a partially transposed $4 \otimes 4$ lattice Bell-diagonal state, and then establish the condition on the eigenspectrum under which the state remains PPT under all permutations of its coefficient vector. We refer to such a state as a $4 \otimes 4$ Bell-diagonal \(\PPPT\) state.

\subsection{Condition for two-ququart Bell-diagonal \(\PPPT\) state}
Let the ququart systems \( A=A_1A_2\) and \(B=B_1B_2\), where each \(A_i\) and \(B_i\) is a qubit. Let \(\{\ket{\Phi_\mu}\}_{\mu\in\mathcal \F_2^2}\) denote the two-qubit Bell basis, and write the corresponding projector \( P_\mu:=\ket{\Phi_\mu}\!\bra{\Phi_\mu}. \) As defined in Eq.~\eqref{eq:ququart_Bell}, the two-ququart lattice Bell states can be expressed as
\begin{equation}
\ket{B_{\mu\nu}}_{AB} = S\left( \ket{\Phi_\mu}_{A_1B_1} \otimes \ket{\Phi_\nu}_{A_2B_2} \right), \;\; \mu,\nu\in\F_2^2,
\label{eq:lattice-Bell-product}
\end{equation}
where $\ket{\Phi_x}_{A_iB_i}=I_2^{A_i}\otimes Q_x^{B_i}\ket{\Phi^+}_{A_iB_i}$ with $x\in\F_2^2$, and \(S\) implements the canonical subsystem permutation from the interleaved ordering \(A_1B_1A_2B_2\) to the bipartite split \(A_1A_2B_1B_2\). The corresponding projectors are \( P_{\mu\nu}^{AB}=\ket{B_{\mu\nu}}_{AB}\bra{B_{\mu\nu}},\)
and an arbitrary lattice Bell-diagonal state takes the form \[\rho(c)^{AB}=\sum_{\mu,\nu\in\F_2^2}c_{\mu\nu}P_{\mu\nu}^{AB}.\] Let $c_{(1)} \ge c_{(2)} \ge \cdots \ge c_{(16)}$ denote the components of $c$ arranged in descending order. The following Proposition provides a complete characterization of the $4\otimes 4$ \(\PPPT\) lattice Bell-diagonal states.
\begin{proposition}
\label{prop:Perm_PPT}
 A two-ququart lattice Bell-diagonal state is \(\PPPT\) if and only if its Bell coefficients, arranged in decreasing order $c_{(1)}\geq\cdots\geq c_{(16)}$, satisfy $\sum_{i=1}^{6}c_{(i)}\leq1/2$.
\end{proposition}
The proof of Proposition~\ref{prop:Perm_PPT} is given in Appendix~\ref{app:lemma1}. Equivalently, we refer to such $c$ as a \(\PPPT\) coefficient vector, and denote the set of all such vectors by $\mathcal{C}$. Mathematically, the \(\PPPT\) set $\mathcal{C}$ is defined by
\begin{eqnarray}
\label{eq:C}
 \mathcal C=  \left\{c\in\mathbb R^{16}:
 c_{\mu\nu}\geq0,\ \sum_{\mu,\nu\in\F_2^2}c_{\mu\nu}=1,~\sum_{i=1}^6c_{(i)}\le\frac12\right\}.\nonumber\\
\end{eqnarray}
Note that the set $\mathcal C$ is invariant under all coordinate permutations.

\subsection{Complete vertex characterization of \texorpdfstring{$\mathcal{C}$}{C}}
\label{sec:vertices_C}
To characterize the \(\PPPT\) coefficient vectors $c$, we first label their coordinates by the decimal integers $\{1,2,\dots,16\}$. Let $n$ denote the number of non-zero entries in $c$, with $z = 16 - n$ zeros. The following lemma establishes a lower bound on $n$ that is necessary for a coefficient vector $c$ to remain in the \(\PPPT\) set $\mathcal C$.
\begin{lemma}
\label{lem:perm_ppt_support}
Every \(\PPPT\) coefficient vector $c \in \mathcal{C}$ contains at least twelve non-zero entries, i.e., $n \ge 12$ (or equivalently, $z \le 4$).
\end{lemma}
\begin{proof}
Assume the non-zero entries are sorted in descending order, with $\sum_{i=1}^n c_{(i)} = 1$ and $s := \sum_{i=1}^6 c_{(i)} \le 1/2$. Since $c_{(i)}$ is non-increasing, $c_{(6)} \le \frac{s}{6}$, and therefore $1 - s \le \frac{(n - 6)s}{6}.$ This yields $s \ge 6/n$. Since $s \le 1/2$, it follows that $n \ge 12$ and $z = 16 - n \le 4$.
\end{proof}
The following lemma provides a complete characterization of the extreme points of the \(\PPPT\) polytope $\mathcal{C}$. It shows that every vertex is a distinct coordinate permutation of one of the five explicit spectral profiles, indexed by the number $z$ of vanishing coefficients, with $z\le4$. This reduces the separability of all states in $\mathcal C$ to a finite set of vertices.
\begin{lemma}[Vertices of \(\mathcal C\)]
\label{thm:vertices}
The vertices of \(\mathcal C\) are exactly the distinct coordinate permutations of \(v(z)\), \(z=0,1,2,3,4\), where 
\begin{equation}
\label{eq:vz}
 v(z)=\left(a_z, ~b_z^{\times (15-z)},~0^{\times z}\right),
\end{equation}
and
\begin{equation}
\label{eq:ab}
 b_z=\frac1{2(10-z)},\qquad
 a_z=\frac{5-z}{2(10-z)}.
\end{equation}
\end{lemma}
We refer the reader to Appendix~\ref{app:lemma1} for the proof of the Lemma~\ref{thm:vertices}. Below, we list all vertices of $\mathcal C$ along with their multiplicities.
\begin{corollary}[Number of vertices]
\label{cor:11036}
The polytope $\mathcal{C}$ possesses exactly $11036$ extreme points, partitioned across five permutation classes indexed by the number of vanishing components $z \in \{0, 1, 2, 3, 4\}$:
\begin{center}
\setlength{\tabcolsep}{7pt}
\renewcommand{\arraystretch}{1.75}
\begin{tabular}{c c c c}
\toprule
\(z\) & Spectral profile & Zero count & Multiplicity\\
\midrule
$0$ & \(\left(a_0,b_0^{\times15}\right)\) & $0$ & \(16\)\\
$1$ & \(\left(a_1,b_1^{\times14},0\right)\) & $1$ & \(16\binom{15}{1}=240\)\\
$2$ & \(\left(a_2,b_2^{\times13},0^{\times2}\right)\) & $2$ & \(16\binom{15}{2}=1680\)\\
$3$ & \(\left(a_3,b_3^{\times12},0^{\times3}\right)\) & $3$ & \(16\binom{15}{3}=7280\)\\
$4$ & \(\left((1/12)^{\times12},0^{\times4}\right)\) & $4$ & \(\binom{16}{4}=1820\)\\
\midrule
\multicolumn{3}{l}{\textbf{Total vertices}} & $\mathbf{11036}$ \\
\bottomrule
\end{tabular}
\end{center}
\end{corollary}
\subsection{Separability of two-ququart Bell-diagonal \(\PPPT\) set}
We now show that the $11036$ vertices decompose into $13$ inequivalent orbits under the action of the affine symplectic group $\operatorname{ASp}(4,2)$. As in the two-qutrit setting, the symmetries in $\operatorname{ASp}(4,2)$ preserve the entanglement properties of the corresponding states. The action of the affine symplectic group on the two-ququart lattice Bell simplex is analyzed in detail in Appendix~\ref{sec:affinesymm_4x4}. Consequently, certifying the separability of a single target representative from each of these $13$ orbits is sufficient to establish that every vertex of $\mathcal{C}$ is associated with a separable two-ququart lattice Bell-diagonal state.
\subsubsection{Thirteen affine-symplectic orbits}
\label{subsec:ququart-separability}
To determine the orbit decomposition of the $11036$ vertices under the action of $\mathrm{ASp}(4,2)$, we analyze the possible placements of the coefficients in each of the five spectral types \(v(z)\) from Corollary~\ref{cor:11036} according to their affine symplectic geometry. 

For \(z\leq 3\), one has \(a_z>b_z\), so a vertex is specified by
(i) a distinguished label \(d\in\mathbb F_2^4\) carrying the coefficient
\(a_z\), (ii) an unordered set \(Z\subset\mathbb F_2^4\setminus\{d\}\)
of \(z\) zero positions, and (iii) the coefficient \(b_z\) at every
remaining position. We therefore write
\begin{equation}
    v_{d,Z}^{(z)}(x)=
\begin{cases}
a_z, & x=d,\\
0,   & x\in Z,\\
b_z, & \text{otherwise},
\end{cases}
\qquad x\in\mathbb F_2^4,
\label{eq:vdz}
\end{equation}
where \(a_z\) and \(b_z\) are given in Eq.~\eqref{eq:ab}. For \(z=4\),
one has \(a_4=b_4=1/12\), so there is no distinguished label \(d\); the
vertex is completely determined by its four-element zero set
\(Z\subset\mathbb F_2^4\).

Consider an affine symplectic transformation \( g(x)=Sx+t, \; S\in\operatorname{Sp}(4,2),\; t\in\mathbb F_2^4. \) For \(z\leq3\), choosing \(t=Sd\) gives \( g(d)=0\) and \(g(x)=S(x+d),\) where we used the fact that addition and subtraction coincide over \(\mathbb F_2\). Consequently, after moving the distinguished label to the zero-th position, the classification reduces to the action of \(\operatorname{Sp}(4,2)\) on the translated zero set \(Z+d:=\{x+d:x\in Z\}.\) The symplectic products of the corresponding displacement vectors are invariant:
\begin{equation}
\label{eq:difference-invariant}
\bigl[g(x)+g(d),\,g(y)+g(d)\bigr] = [x+d,\,y+d].
\end{equation}
Together with the linear-dependence relations among the translated vectors, these pairwise symplectic products completely determine the relevant orbit types. More precisely, let
\(\mathcal U=\{u_1,\ldots,u_r\}\) and
\(\mathcal U'=\{u_1',\ldots,u_r'\}\) be two translated configurations, suitably relabeled. Suppose that
\[
 \sum_{i=1}^{r}\alpha_i u_i=0
 \quad\Longleftrightarrow\;
 \sum_{i=1}^{r}\alpha_i u_i'=0
\]
for every \(\alpha_i\in\F_2,\) and that
\(
 [u_i,u_j]=[u_i',u_j']
 ~\forall ~i,j.
\)
The first condition ensures that the assignment \(u_i\mapsto u_i'\) induces a well-defined linear map
\[
 T:\operatorname{span}(\mathcal U)
 \longrightarrow
 \operatorname{span}(\mathcal U').
\]
The second condition, together with bilinearity, implies
\[
 [T\xi,T\eta]=[\xi,\eta]
 \qquad
 \text{for all }
 \xi,\eta\in\operatorname{span}(\mathcal U),
\]
so \(T\) is a symplectic isometry between the two spans. Witt's extension theorem then guarantees that \(T\) can be extended to a global symplectic transformation
\(
 S\in\Sp(4,2)
\)
acting on all of \(\F_2^4\) \cite[Theorem~7.4]{Taylor1992}. Consequently, \(S(\mathcal U)=\mathcal U'\); after restoring the affine translations, the original configurations belong to the same \(\ASp(4,2)\)-orbit. Thus, two translated configurations of zero-coefficient labels belong to the same \(\operatorname{Sp}(4,2)\)-orbit if and only if, after suitably matching their elements, they have the same linear-dependence relations and pairwise symplectic products.

The detailed representative counting arguments for \(z=0,1,2,3,4\) are given in Appendix~\ref{app:lattice-orbits_4x4}. In Appendix~\ref{app:numerical_algorithm}, following the approach used in the two-qutrit setting, we briefly describe the numerical algorithm to systematically extract both the orbit sizes and their canonical representatives. The numbers of affine symplectic orbits arising from the five spectral types are, respectively, \[ (1,1,2,5,4), \] giving \(13\) distinct orbits in total. Grouped by the zero count \(z\), their sizes are
\[
\begin{aligned}
z=0 &: \quad \{16\},\\
z=1 &: \quad \{240\},\\
z=2 &: \quad \{960,720\},\\
z=3 &: \quad \{320,2880,2880,960,240\},\\
z=4 &: \quad \{80,1440,240,60\}.
\end{aligned}
\]
The sum of these orbit sizes is \(11036\).
Table~\ref{tab:lattice-orbits} records the corresponding representatives, and orbit sizes. The classification is based on the ranks, linear-dependence relations, and symplectic products of suitable difference sets associated with the zero positions.

\begin{table}[htbp]
\setlength{\tabcolsep}{10pt}
\small
\renewcommand{\arraystretch}{1.5}
\centering
\caption{The resulting thirteen affine-symplectic lattice orbits are summarized below, adopting the notation $v_{d,Z}^z$ introduced in Eq.~\eqref{eq:vdz}. As an illustrative example, the representative of the fifth orbit is given by the coefficient vector $(a_3, 0^{\times 3}, b_3^{\times 12})$, where $a_3 = 1/7$ and $b_3 = 1/14$ (Eq.~\eqref{eq:ab}).}
\label{tab:lattice-orbits}
\begin{tabular}{c c c c r}
\toprule
orbit&\(z\)&\(d\)&\(Z\)&multiplicity\\
\midrule
1&0&\(\hex{0}\)&\(\varnothing\)&16\\
2&1&\(\hex{0}\)&\(\{\hex{1}\}\)&240\\
3&2&\(\hex{0}\)&\(\{\hex{1},\hex{2}\}\)&960\\
4&2&\(\hex{0}\)&\(\{\hex{1},\hex{4}\}\)&720\\
5&3&\(\hex{0}\)&\(\{\hex{1},\hex{2},\hex{3}\}\)&320\\
6&3&\(\hex{0}\)&\(\{\hex{1},\hex{2},\hex{4}\}\)&2880\\
7&3&\(\hex{0}\)&\(\{\hex{1},\hex{2},\hex{5}\}\)&2880\\
8&3&\(\hex{0}\)&\(\{\hex{1},\hex{2},\hex{7}\}\)&960\\
9&3&\(\hex{0}\)&\(\{\hex{1},\hex{4},\hex{5}\}\)&240\\
10&4&--&\(\{\hex{0},\hex{1},\hex{2},\hex{3}\}\)&80\\
11&4&--&\(\{\hex{0},\hex{1},\hex{2},\hex{4}\}\)&1440\\
12&4&--&\(\{\hex{0},\hex{1},\hex{2},\hex{7}\}\)&240\\
13&4&--&\(\{\hex{0},\hex{1},\hex{4},\hex{5}\}\)&60\\
\midrule
\multicolumn{4}{r}{total}&11036\\
\bottomrule
\end{tabular}
\end{table}
\subsubsection{Separability of two-ququart lattice Bell-diagonal states associated to affine-symplectic orbits}
It now remains to establish the separability of the two-ququart lattice Bell-diagonal states corresponding to the thirteen affine-symplectic orbits. To accomplish this, let us identify the sixteen elements of \(\mathbb F_2^4\) with the hexadecimal symbols \( \mathtt{0},\mathtt{1},\ldots,\mathtt{9}, \mathtt{A},\ldots,\mathtt{F}, \)
according to \( \mathtt{0}\leftrightarrow0000,\; \mathtt{1}\leftrightarrow0001,\;\ldots,\; \mathtt{F}\leftrightarrow1111.\)
If an affine Lagrangian plane is \(L=\{a,b,c,d\}\)\footnote{Since every affine Lagrangian plane \(L\subset V_4=\mathbb F_2^4\) contains four points, we write
\(L=\{a,b,c,d\},\) where \(a,b,c,d\in\mathbb F_2^4.\)}, we denote it by the hexadecimal code \(\ell=\mathtt{abcd}\), with its elements written in increasing hexadecimal order, and define
\[
\sigma_{\ell} := \frac14\left(P_a+P_b+P_c+P_d\right).
\]
For example,
\[
\sigma_{\mathtt{45CD}} = \frac14\left( P_{0100}+P_{0101}+P_{1100}+P_{1101} \right).
\]
For each orbit representative \(\rho_r\), let
\(
\mathcal S_r = \left(\ell_{r,1},\ldots,\ell_{r,m_r}\right)
\)
denote the ordered list of hexadecimal plane codes appearing in its decomposition. The corresponding rational weights satisfy
\begin{equation}
\label{eq:sep_decom_4x4}
    \rho_r = \sum_{j=1}^{m_r} q_{rj}\sigma_{\ell_{r,j}}, \;\; q_{rj}>0, \;\; \sum_{j=1}^{m_r}q_{rj}=1.
\end{equation}
Thus, every entry of \(\mathcal S_r\) directly specifies the four lattice Bell projectors in the corresponding separable plane state. The supports and exact rational weights are given in Table~\ref{tab:lattice-decomposition-certificates}, while the orbit representatives are given in Table~\ref{tab:lattice-orbits}. Since, by Lemma~\ref{lemma:Coset_separability}, every $\sigma_{\ell_{r,j}}$ is separable, the decompositions in Eq.~\eqref{eq:sep_decom_4x4} establish the separability of all thirteen representatives, and hence of their complete affine-symplectic orbits. The accompanying \textit{Wolfram Mathematica} notebook~\cite{SupplementalMaterial} independently verifies all identities exactly.

\begin{table*}[t]
\caption{Exact convex decompositions of the thirteen affine-symplectic
orbit representatives into separable affine Lagrangian-plane states.
Here \(\sigma_{\mathtt{abcd}}=\frac14(P_a+P_b+P_c+P_d)\), where the hexadecimal symbols \(\mathtt{0},\ldots,\mathtt{F}\) represent the corresponding four-bit
labels in \(\mathbb F_2^4\) according to \( \mathtt{0}\leftrightarrow0000,\; \mathtt{1}\leftrightarrow0001,\;\ldots,\; \mathtt{F}\leftrightarrow1111\). The thirteen representatives $\{\rho_r\}_{r=1}^{13}$ are listed in Table~\ref{tab:lattice-orbits}. All the displayed identities are exact over the rationals. The \textit{Wolfram Mathematica} notebook verifying them is provided in the Supplementary Material~\cite{SupplementalMaterial}.}
\label{tab:lattice-decomposition-certificates}
\centering
\small
\setlength{\tabcolsep}{15pt}
\renewcommand{\arraystretch}{2.25}
\begin{tabular}{c c c p{12.0cm}}
\toprule
\(r\) & \(z\) & orbit & exact weights \(q_{rj}\)\\
\midrule
1 & 0 & 16 & \(\displaystyle\rho_1=\frac15\bigl( \sigma_{\mathtt{0145}} +\sigma_{\mathtt{02CE}} +\sigma_{\mathtt{038B}} +\sigma_{\mathtt{069F}} +\sigma_{\mathtt{07AD}} \bigr). \)\\
2 & 1 & 240 & \(\displaystyle\rho_2=\frac1{18}\bigl(\sigma_{\mathtt{028A}}
+\sigma_{\mathtt{038B}}+\sigma_{\mathtt{069F}}+\sigma_{\mathtt{079E}}+\sigma_{\mathtt{4589}}+\sigma_{\mathtt{89CD}}\bigr)+\frac16\bigl(\sigma_{\mathtt{02CE}}+\sigma_{\mathtt{0347}}+\sigma_{\mathtt{05AF}}+\sigma_{\mathtt{06BD}}\bigr).\)\\
3 & 2 & 960 & \(\displaystyle\rho_3=\frac18\bigl(\sigma_{\mathtt{0347}}+\sigma_{\mathtt{03CF}}+\sigma_{\mathtt{05AF}}+\sigma_{\mathtt{079E}}+\sigma_{\mathtt{4589}}+\sigma_{\mathtt{8ACE}}\bigr)+\frac14\sigma_{\mathtt{06BD}}.\)\\
4 & 2 & 720 & \(\displaystyle\rho_4=\frac18\bigl(\sigma_{\mathtt{028A}}+\sigma_{\mathtt{03CF}}+\sigma_{\mathtt{069F}}+\sigma_{\mathtt{07AD}}+\sigma_{\mathtt{2367}}+\sigma_{\mathtt{89CD}}\bigr)+\frac14\sigma_{\mathtt{05BE}}.\)\\
5 & 3 & 320 & \(\displaystyle\rho_5=\frac17\bigl(\sigma_{\mathtt{05BE}}+\sigma_{\mathtt{069F}}+\sigma_{\mathtt{079E}}+\sigma_{\mathtt{07AD}}+\sigma_{\mathtt{45CD}}+\sigma_{\mathtt{468A}}+\sigma_{\mathtt{8BCF}}\bigr).\)\\
6 & 3 & 2880 & \(\displaystyle\rho_6=\frac17\bigl(\sigma_{\mathtt{03CF}}+\sigma_{\mathtt{069F}}+\sigma_{\mathtt{06BD}}+\sigma_{\mathtt{07AD}}+\sigma_{\mathtt{358E}}+\sigma_{\mathtt{579B}}+\sigma_{\mathtt{8ACE}}\bigr).\)\\
7 & 3 & 2880 & \(\displaystyle\rho_7=\frac17\bigl(\sigma_{\mathtt{0347}}+\sigma_{\mathtt{03CF}}+\sigma_{\mathtt{06BD}}+\sigma_{\mathtt{079E}}+\sigma_{\mathtt{468A}}+\sigma_{\mathtt{8BCF}}+\sigma_{\mathtt{9ADE}}\bigr).\)\\
8 & 3 & 960 & \(\displaystyle\rho_8=\frac27\sigma_{\mathtt{038B}}+\frac17\bigl(\sigma_{\mathtt{05AF}}+\sigma_{\mathtt{069F}}+\sigma_{\mathtt{45CD}}+\sigma_{\mathtt{46CE}}+\sigma_{\mathtt{9ADE}}\bigr).\)\\
9 & 3 & 240 & \(\displaystyle\rho_9=\frac17\bigl(\sigma_{\mathtt{028A}}+\sigma_{\mathtt{03CF}}+\sigma_{\mathtt{06BD}}+\sigma_{\mathtt{079E}}+\sigma_{\mathtt{279C}}+\sigma_{\mathtt{368D}}+\sigma_{\mathtt{ABEF}}\bigr).\)\\
10 & 4 & 80 & \(\displaystyle\rho_{10}=\frac16\bigl(\sigma_{\mathtt{478B}}+\sigma_{\mathtt{47CF}}
+\sigma_{\mathtt{569A}}+\sigma_{\mathtt{56DE}}+\sigma_{\mathtt{8BCF}}+\sigma_{\mathtt{9ADE}}\bigr).\)\\
11 & 4 & 1440 & \(\displaystyle\rho_{11}=\frac16\bigl(\sigma_{\mathtt{358E}}+\sigma_{\mathtt{35AC}}+\sigma_{\mathtt{67AB}}+\sigma_{\mathtt{67EF}}+\sigma_{\mathtt{89CD}}+\sigma_{\mathtt{9BDF}}\bigr).\)\\
12 & 4 & 240 & \(\displaystyle\rho_{12}=\frac16\bigl(\sigma_{\mathtt{349E}}+\sigma_{\mathtt{368D}}+\sigma_{\mathtt{45CD}}+\sigma_{\mathtt{569A}}+\sigma_{\mathtt{8BCF}}+\sigma_{\mathtt{ABEF}}\bigr).\)\\
13 & 4 & 60 & \(\displaystyle\rho_{13}=\frac16\bigl(\sigma_{\mathtt{23EF}}+\sigma_{\mathtt{278D}}+\sigma_{\mathtt{369C}}+\sigma_{\mathtt{67AB}}+\sigma_{\mathtt{8ACE}}+\sigma_{\mathtt{9BDF}}\bigr).\)\\
\bottomrule
\end{tabular}
\end{table*}

\begin{theorem}
\label{thm:lattice-main}
Let $\rho(c)$ be a two-ququart lattice Bell-diagonal state. If $\rho(\pi c)^{\Gamma_A} \ge 0$ for all permutations $\pi \in S_{16}$, then $\rho(c)$ is separable. Consequently, every permuted state $\rho(\pi c)$ is also separable.
\end{theorem}
\begin{proof}
By Eqs.~\eqref{eq:lattice-permutation-PPT-largest-six} and~\eqref{eq:C}, the hypothesis is equivalent to $c \in \mathcal{C}$. By Corollary~\ref{cor:11036}, the polytope $\mathcal{C}$ is the convex hull of its $11036$ extreme points, so $c$ decomposes as a convex combination of these vertices. By Eq.~\eqref{eq:lattice-projector-action}, every vertex is locally Clifford equivalent to one of the thirteen canonical orbit representatives, each of which is certified separable by the exact decomposition in Eq.~\eqref{eq:sep_decom_4x4}. Because local unitaries and convex mixtures preserve separability, every vertex corresponds to a separable state, implying that every state in $\mathcal{C}$ is separable. Finally, because $\mathcal{C}$ is invariant under all coordinate permutations $\pi \in S_{16}$, the same separability guarantee applies to every permuted state $\rho(\pi c)$.
\end{proof}

Taking the contrapositive of Theorem~\ref{thm:lattice-main} yields the operationally useful characterization
\begin{equation}
\label{eq:lattice-contrapositive}
\rho(c) \text{ is entangled} \implies \rho(\pi c) \text{ is NPT for some } \pi \in S_{16}.
\end{equation}
While PPT bound-entangled lattice states do exist~\cite{BenattiLiguori2007,BenattiKarbalaii2013}, Eq.~\eqref{eq:lattice-contrapositive} reveals that no such state can remain PPT under all permutations of its Bell coefficients. Consequently, any entangled two-ququart lattice Bell-diagonal state is necessarily non-$\APPT$, leading directly to the following main result:
\begin{theorem}
Every absolutely PPT two-ququart lattice Bell-diagonal state is separable. Equivalently,
\begin{equation}
\APPT_{4,4}\cap\BD(4)\subseteq\SEP .
\end{equation}
\end{theorem}
\begin{proof}
Because absolute PPT requires PPT preservation under all global unitaries, it applies in particular to the Bell-basis permutation unitaries implementing arbitrary phase-space relabeling \(\pi\in S_{16}\). Hence every $\APPT$ lattice Bell-diagonal state is \(\PPPT\). By Theorem~\ref{thm:lattice-main}, every state in the \(\PPPT\) polytope $\mathcal{C}$ is separable, thereby completing the proof.
\end{proof}
\section{Conclusion}
\label{sec:conclusion}

In this work, we have investigated whether absolutely PPT ($\APPT$) states can be entangled within two natural Bell-diagonal families. We have established that every $\APPT$ state diagonal in the standard two-qutrit Weyl--Heisenberg Bell basis, or in the two-ququart lattice Bell basis, is separable. In fact, our separability results extend to larger classes of states than the $\APPT$ states alone.

For a Bell-diagonal state $\rho(c)=\sum_i c_iP_i$, the Bell coefficients $c_i\geq0$, with $\sum_i c_i=1$, are its eigenvalues; let $c_{(1)}\geq\cdots\geq c_{(d^2)}$ denote their decreasing rearrangement. For the two-qutrit Weyl--Heisenberg Bell basis, we prove that the inequality $c_{(1)}\leq c_{(7)}+c_{(8)}+c_{(9)}$ is sufficient for separability. Since every $\APPT$ spectrum satisfies this inequality, every $\APPT$ state diagonal in this basis is separable. The inequality also provides a directly checkable sufficient criterion for separability for states beyond the $\APPT_{3,3}$ set, and it guarantees that the state remains separable even under the conjugation of every global unitary that permutes the Bell coefficients among the Bell states of this basis.

For the $4\otimes4$ lattice basis, we proved that every entangled lattice Bell-diagonal state becomes NPT after a suitable permutation of its Bell coefficients. Equivalently, every \(\PPPT\) lattice Bell-diagonal state is separable. We demonstrated that the \(\PPPT\) region is exactly the set defined by the spectral inequality $\sum_{i=1}^{6}c_{(i)}\leq1/2$, and established its separability through affine-symplectic orbit reduction and exact convex decompositions into uniform affine Lagrangian-coset states. Since each coefficient permutation is implemented by a global unitary, every entangled lattice Bell-diagonal state fails to be $\APPT$.

Our results establish the separability of all $\APPT$ states diagonal in the two Bell-diagonal families considered here. Absolute separability, however, additionally requires the separability of every global-unitary conjugate, including conjugates that lie outside the corresponding Bell simplex. The present results therefore provide basis-specific evidence about how spectral PPT conditions can force separability, while the distinction between separability and absolute separability remains explicit.

A natural direction for future work is to extend the present approach to other Bell-basis constructions, including generalized Bell-diagonal families obtained by rotating the standard Weyl--Heisenberg Bell basis by diagonal unitaries~\cite{Popp_NJP_2024}. Other natural settings include the standard ququart Weyl--Heisenberg basis, as well as its extensions to higher dimensions. Applying spectral inequalities, permutation arguments, symmetry reduction, and explicit separable decompositions to these families may yield new spectrum-based separability criteria and help identify broader classes in which \(\PPPT\) implies separability.

\appendix
\section{Proof of Lemma~\ref{lemma:Coset_separability}, Lemma~\ref{lemma:rank_vertices_outerpolytope}, Proposition~\ref{prop:Perm_PPT}, and Lemma~\ref{thm:vertices}}
\label{app:lemma1}
For the reader's convenience, we recall the statement of each lemma and proposition used in the main text.\\

\noindent
\textbf{Lemma 1} (Lagrangian-coset separability). \textit{Let \(\mathcal L\) be a Lagrangian subspace for either the qutrit Weyl system or the binary two-qubit Pauli system. Thus, \(|\mathcal L|=d,\) where \(d=3\) for the qutrit system and \(d=4\) for the binary two-qubit system. For every affine coset \(L=t+\mathcal L\), the Bell-diagonal state
\[
\sigma_L=\frac{1}{d}\sum_{x\in L}P_x
\]
is separable.}
\begin{proof}
The qutrit line-state construction and the four-point lattice construction are standard in the magic-simplex and lattice-state literature~\cite{Baumgartner2007,Popp2022,BenattiKarbalaii2013}; for completeness, we give a uniform self-contained proof covering both settings by directly applying group theory. Let \(U_x\) denote the Weyl or Pauli operator associated with the label \(x\), so that \(P_x=(I\otimes U_x)P_0(I\otimes U_x^\dagger),\;P_0=\ket{\Omega_d}\!\bra{\Omega_d},\) where \(\ket{\Omega_d} = \frac{1}{\sqrt d}\sum_{r=1}^{d}\ket{r}\otimes\ket{r}\) is the standard maximally entangled state.

By the defining property of a Lagrangian subspace, the symplectic product vanishes between every pair of its elements:
\[[h,h']=0\qquad\text{for all }h,h'\in \mathcal L.\]
The Weyl--Pauli commutation relation therefore implies \[ U_hU_{h'}=U_{h'}U_h, \qquad h,h'\in \mathcal L.\] Thus, \(\mathcal{U}_{\mathcal L}=\{U_h:h\in \mathcal L\}\) is a commuting family of unitary operators and hence possesses a common orthonormal eigenbasis \( \{\ket{e_j}\}_{j=1}^{d}. \) We may consequently write
\[
U_h\ket{e_j}=\chi_j(h)\ket{e_j}.
\]
Under addition, \(\mathcal L\) is a finite abelian group of order \(d\). After choosing compatible phases for the operators \(U_h\), their restriction to \(\mathcal L\) satisfies \[ U_{h+h'}=U_hU_{h'}.\] It follows that \[ \chi_j(h+h')=\chi_j(h)\chi_j(h'), \]
so \(h\mapsto\chi_j(h)\) is the character of $j$-th one-dimensional irreducible representation of the abelian group \(\mathcal L\). Moreover, the Lagrangian subgroup gives \(d\) distinct joint characters \(\chi_1,\ldots,\chi_d\). Now, define the dephasing channel
\[
\mathcal{D}_{\mathcal L}(A):=\frac{1}{d}\sum_{h\in \mathcal L}U_hAU_h^\dagger,
\]
for some bounded operator $A$. Expanding \(A\) in the common eigenbasis as \(A=\sum_{j,k=1}^{d}A_{jk}\ket{e_j}\!\bra{e_k},\) we obtain
\[ U_hAU_h^\dagger = \sum_{j,k=1}^{d} A_{jk}\chi_j(h)\chi_k(h)^* \ket{e_j}\!\bra{e_k}.\]
Therefore,
\[\mathcal{D}_H(A) = \sum_{j,k=1}^{d} A_{jk} \left[ \frac{1}{d}\sum_{h\in H} \chi_j(h)\chi_k(h)^*\right]\ket{e_j}\!\bra{e_k}.\]
Orthogonality of the characters of the finite abelian group \(\mathcal L\) gives
\[\frac{1}{d}\sum_{h\in \mathcal L} \chi_j(h)\chi_k(h)^* = \delta_{jk}. \]
Hence
\[
\mathcal{D}_{\mathcal L}(A)
=
\sum_{j=1}^{d}
\ket{e_j}\!\bra{e_j}
A
\ket{e_j}\!\bra{e_j}.
\]
Thus, \(\mathcal{D}_{\mathcal L}\) is the complete-dephasing channel in the common eigenbasis \(\{\ket{e_j}\}\). Now, we first consider the Bell-diagonal state associated with \(\mathcal L\):
\[
\begin{aligned}
\sigma_{\mathcal L}
&=
\frac{1}{d}\sum_{h\in \mathcal L}P_h\\
&=
\frac{1}{d}\sum_{h\in \mathcal L}
(I\otimes U_h)P_0(I\otimes U_h^\dagger)\\
&=
(I\otimes\mathcal{D}_{\mathcal L})(P_0).
\end{aligned}
\]
The maximally entangled state has the basis-independent representation
\(\ket{\Omega_d}=\frac{1}{\sqrt d}\sum_{j=1}^{d}\ket{e_j^*}\otimes\ket{e_j},\)
where complex conjugation is taken in the reference computational basis. Applying \(\mathcal{D}_{\mathcal L}\) to the second subsystem of \(P_0=\ketbra{\Omega_d}{\Omega_d}\) removes all terms with \(j\neq k\), and therefore
\[
\sigma_{\mathcal L}
=
\frac{1}{d}\sum_{j=1}^{d}
\ket{e_j^*}\!\bra{e_j^*}
\otimes
\ket{e_j}\!\bra{e_j}.
\]
Hence \(\sigma_{\mathcal L}\) is separable.

Finally, consider an arbitrary affine coset \(L=t+\mathcal L\). Both the qutrit Weyl operators and the two-qubit Pauli operators obey the projective label-addition rule \( U_{t+h}=e^{i\vartheta(t,h)}U_tU_h. \)
The phase \(e^{i\vartheta(t,h)}\) disappears at the projector level, so \( P_{t+h} = (I\otimes U_t)P_h(I\otimes U_t^\dagger).\) It follows that
\[
\begin{aligned}
\sigma_{t+\mathcal L}
=
\frac{1}{d}\sum_{h\in \mathcal L}P_{t+h}
=
(I\otimes U_t)\sigma_{\mathcal L}(I\otimes U_t^\dagger).
\end{aligned}
\]
Thus, \(\sigma_{t+\mathcal L}\) is locally unitarily equivalent to \(\sigma_{\mathcal L}\). Since local unitary transformations preserve separability, \(\sigma_{t+{\mathcal L}}\) is also separable.
\end{proof}

\noindent
\textbf{Lemma 3} (Rank characterization of vertices). \textit{Let $\lambda = \sum_{k=1}^9 \alpha_k u_k \in Q_9^\downarrow$ with $\alpha_k \ge 0$ and $\sum_{k=1}^9 \alpha_k = 1$. Defining the barycentric rank of $\lambda$, denoted by $\operatorname{rank}(\lambda)$, as the number of non-zero coefficients in its barycentric decomposition, the following conditions hold:
    \begin{enumerate}
    \item If $h_9(\lambda) > 0$ and $\operatorname{rank}(\lambda) \ge 2$, then $\lambda$ is not a vertex of $Q_9^\downarrow$.
    \item If $h_9(\lambda) = 0$ and $\operatorname{rank}(\lambda) \ge 3$, then $\lambda$ is not a vertex of $Q_9^\downarrow$.
    \end{enumerate}}
\begin{proof}
    Let $\lambda = \sum_{k=1}^9 \alpha_k u_k = \alpha \cdot u$, where we adopt the vector notation $\alpha = (\alpha_1, \dots, \alpha_9)^T$ and $u = (u_1, \dots, u_9)^T$. We now establish the two conditions in turn.
    
    (1.) Suppose $h_9(\lambda) > 0$ and $\operatorname{rank}(\lambda) \ge 2$. There exist distinct indices $i \ne j$ such that $\alpha_i, \alpha_j > 0$. Consider the perturbation direction $d = e_i - e_j \in \mathbb{R}^9$, where $e_k$ denotes the $k$-th canonical basis vector in the $9$-th dimensional coefficient space. For example, $e_i=(0,...,0,1,0,...,0)$, with the $1$ in the $i$-th position, and similarly for $e_j$. For sufficiently small $\varepsilon > 0$, the perturbed coefficients $\alpha^\pm := \alpha \pm \varepsilon d$ remain non-negative and satisfy $\sum_{k=1}^9 \alpha_k^\pm = 1$. Because $h_9$ is linear, the spectra $\lambda^\pm := \sum_{k=1}^9 \alpha_k^\pm u_k$ evaluate to
    \begin{equation*}
    h_9(\lambda^\pm) = h_9(\lambda) \pm \varepsilon \bigl(h_9(u_i) - h_9(u_j)\bigr).
    \end{equation*} 
    Since $h_9(\lambda) > 0$, choosing $\varepsilon > 0$ sufficiently small guarantees $h_9(\lambda^\pm) > 0$, placing both $\lambda^+$ and $\lambda^-$ strictly within $Q_9^\downarrow$. As $\lambda = \frac{1}{2}(\lambda^+ + \lambda^-)$ with $\lambda^+ \ne \lambda^-$, $\lambda$ is not an extreme point.\\
    
    (2.) Suppose $h_9(\lambda) = 0$ and $\operatorname{rank}(\lambda) \ge 3$. Then there exist three distinct indices $i, j, k$ such that $\alpha_i, \alpha_j, \alpha_k > 0$.  Consider a perturbation vector $d \in \mathbb{R}^9$ supported exclusively on $\{i, j, k\}$ that satisfies the homogeneous linear system \[ d_i+d_j+d_k=0, \quad\text{and} \] \[ h(u_i)d_i+h(u_j)d_j+h(u_k)d_k=0. \] Because this constitutes at most two linearly independent constraints in three unknowns, a non-trivial solution $(d_i, d_j, d_k) \ne 0$ always exists.  Setting all other components to zero, we construct $\alpha^\pm := \alpha \pm \varepsilon d$. For small enough $\varepsilon > 0$, $\alpha^\pm \ge 0$, while normalization and $h_9(\lambda^\pm) = 0$ are preserved by construction. Thus, $\lambda^\pm \in Q_9^\downarrow$ are distinct feasible points satisfying $\lambda = \frac{1}{2}(\lambda^+ + \lambda^-)$, demonstrating that $\lambda$ cannot be a vertex.
\end{proof}

\noindent
\textbf{Proposition 3}. \textit{ A two-ququart lattice Bell-diagonal state is \(\PPPT\) if and only if its Bell coefficients, arranged in decreasing order $c_{(1)}\geq\cdots\geq c_{(16)}$, satisfy $\sum_{i=1}^{6}c_{(i)}\leq1/2$.}
\begin{proof}
Recall from Eq.~\eqref{eq:ququart_Bell} that, after identifying
\(A=A_1A_2\) and \(B=B_1B_2\), the lattice Bell basis factorizes into
pairs of two-qubit Bell states:
\[
\ket{B_{\mu\nu}}_{AB} = S\!\left( \ket{\Phi_\mu}_{A_1B_1}\otimes \ket{\Phi_\nu}_{A_2B_2} \right), \;\; \mu,\nu\in\F_2^2,
\]
where \(S\) rearranges \(A_1B_1A_2B_2\) into \(A_1A_2B_1B_2\).
Writing \(P_{\mu\nu}=\ketbra{B_{\mu\nu}}{B_{\mu\nu}}\), a lattice
Bell-diagonal state is
\[
\rho(c)=
\sum_{\mu,\nu\in\F_2^2}c_{\mu\nu}P_{\mu\nu}.
\]
Let \(\Gamma_A=T_{A_1}\otimes I_{B_1}\otimes T_{A_2}\otimes I_{B_2}\) denote partial transposition with respect to the ququart $A$ in the interleaved ordering. The four partially transposed two-qubit Bell projectors possess a common orthonormal eigenbasis. Thus, after a suitable relabeling of the corresponding rank-one eigenprojectors \(\{\Pi_\alpha\}_{\alpha\in\F_2^2}\), one has
\begin{equation}
P_\mu^{\Gamma_{A_i}} = \frac{1}{2} \sum_{\alpha\in\F_2^2} h_{\alpha\mu}\Pi_\alpha,
\;\text{where}\;
h_{\alpha\mu}=1-2\delta_{\alpha\mu}.
\label{eq:Bell-PT-Hadamard}
\end{equation}
Define the rank-one projectors
\[ \Pi_{\alpha\beta}:=S\left(\Pi_\alpha^{A_1B_1}\otimes\Pi_\beta^{A_2B_2}\right)S^\dagger.\]
Using Eq.~\eqref{eq:Bell-PT-Hadamard}, the partial transpose of
\(\rho(c)^{AB}\) becomes
\begin{align}
\rho(c)^{\Gamma_A} = \frac{1}{4} \sum_{\alpha,\beta\in\F_2^2} \left( \sum_{\mu,\nu\in\F_2^2} h_{\alpha\mu}h_{\beta\nu}c_{\mu\nu}
\right)
\Pi_{\alpha\beta}.
\label{eq:lattice-PT-diagonalization}
\end{align}
The sixteen eigenvalues of $\rho(c)^{\Gamma_A}$ are therefore given by
\begin{equation}
\lambda_{\alpha\beta}\!\left(\rho(c)^{\Gamma_A}\right) = \frac{1}{4} \sum_{\mu,\nu\in\F_2^2} h_{\alpha\mu}h_{\beta\nu}c_{\mu\nu}, \;\;\forall \;\alpha,\beta\in\F_2^2.
\label{eq:lattice-PT-eigenvalues-preliminary}
\end{equation}
For each fixed \((\alpha,\beta)\in\F_2^4\), introduce the six-element set
\begin{equation}
X_{\alpha\beta} := \Bigl( \{\alpha\}\times \bigl(\F_2^2\setminus\{\beta\}\bigr) \Bigr) \cup \Bigl( \bigl(\F_2^2\setminus\{\alpha\}\bigr) \times\{\beta\} \Bigr).
\label{eq:lattice-cross-set}
\end{equation}
The product \(h_{\alpha\mu}h_{\beta\nu}\) is negative precisely when
exactly one of the equalities \(\mu=\alpha\) and \(\nu=\beta\) holds. Hence,
\[
h_{\alpha\mu}h_{\beta\nu} =
\begin{cases}
-1, & (\mu,\nu)\in X_{\alpha\beta},\\
+1, & (\mu,\nu)\notin X_{\alpha\beta}.
\end{cases}
\]
It follows that
\begin{align}
\lambda_{\alpha\beta}\!\left(\rho(c)^{\Gamma_A}\right) &=\frac{1}{4}
\left[ \sum_{(\mu,\nu)\notin X_{\alpha\beta}}c_{\mu\nu} - \sum_{(\mu,\nu)\in X_{\alpha\beta}}c_{\mu\nu} \right] \nonumber\\
&= \frac{1}{4} \left[ 1- 2\sum_{(\mu,\nu)\in X_{\alpha\beta}}c_{\mu\nu} \right].
\label{eq:lattice-PT-eigenvalues}
\end{align}
Therefore, $\rho(c)$ is PPT if and only if
\begin{equation*}
\sum_{(\mu,\nu) \in X_{\alpha\beta}} c_{\mu\nu} \le \frac{1}{2} \;\; \text{for all } (\alpha, \beta) \in \mathbb{F}_2^4.
\end{equation*}
Equivalently, \(\rho_c\) is NPT if and only if there exists at least one
\((\alpha,\beta)\) such that
\[
\sum_{(\mu,\nu)\in X_{\alpha\beta}}c_{\mu\nu}
>\frac{1}{2}.
\]
Let $c_{(1)} \ge c_{(2)} \ge \cdots \ge c_{(16)}$ denote the components of $c$ arranged in descending order. Because each set $X_{\alpha\beta}$ consists of exactly six coefficients, an arbitrary permutation $\pi \in S_{16}$ can map the six largest coefficients onto any chosen $X_{\alpha\beta}$. Consequently, we obtain the following necessary and sufficient condition for a $4 \otimes 4$ Bell-diagonal \(\PPPT\) state:
\begin{equation}
\rho(\pi c)\ \text{is PPT for every }\pi\in S_{16}
\;\iff\;
\sum_{i=1}^{6}c_{(i)}\leq\frac{1}{2}.
\label{eq:lattice-permutation-PPT-largest-six}
\end{equation} 
\end{proof}

\noindent
\textbf{Lemma 5} (Vertices of \(\mathcal C\)). \textit{The vertices of \(\mathcal C\) are exactly the distinct coordinate permutations of \(v(z)\), \(z=0,1,2,3,4\), where 
\begin{equation}
 v(z)=\left(a_z, ~b_z^{\times (15-z)},~0^{\times z}\right),
\end{equation}
and
\begin{equation}
\label{eq:ab_app}
 b_z=\frac1{2(10-z)},\qquad
 a_z=\frac{5-z}{2(10-z)}.
\end{equation}}
\begin{proof}
Let $[16] = \{1, 2, \dots, 16\}$, and denote any subset $\bar S \subset [16]$ of cardinality $\vert{}\bar S\vert{} = 6$ as a six-subset. For each such subset, we define the corresponding six-sum by $\sum_{i \in \bar S} c_i$. A six-subset $\bar S$ is termed active on $c$ if it saturates the bound, i.e., $\sum_{i \in \bar S} c_i = 1/2$, and inactive otherwise.
Let us now first note that at least one six-subset must be active at a vertex: otherwise the only equality conditions on the positive support $J(c)=\{i:c_{i}>0\}$ would be normalization, leaving a nontrivial perturbation direction.\\
Set \(b=c_{(6)}\) and
\[
 H=\{i:c_i>b\},\qquad r=|H|\leq5.
\]
Every active six-subset contains every index in \(H\). To see this, suppose an
active \(\bar S\) omitted \(i\in H\). Since \(\bar S\) has six elements and at most five
coefficients exceed \(b\), it contains some \(j\) with \(c_j\leq b<c_i\).
Replacing \(j\) by \(i\) gives a six-sum
\[
 \frac12-c_j+c_i>\frac12,
\]
contradicting feasibility of the vector $c$ to remain inside $\mathcal C$.

If \(r\geq2\), choose distinct \(i,j\in H\) and perturb
\(c_i\mapsto c_i+\varepsilon\), \(c_j\mapsto c_j-\varepsilon\) such that both entries remain inside $H$ even after the perturbation. Every active six-subset contains both entries, so all active equations and normalization remain unchanged. There are finitely many inactive inequalities and each has positive slack, i.e., $\left(\tfrac12-\sum_{i\in \bar S}c_i\right)>0$ corresponding to every inactive six-subset $S$. Therefore, for sufficiently small positive and negative \(\varepsilon\) (more precisely, \(|\varepsilon|\) is less than the smallest positive slack), positivity and all inactive inequalities also remain valid. If we consider that a positive and a negative \(\varepsilon\) yield $c_+$ and $c_-$, respectively, then the point $c$ can be expressed as an equal mixture of $c_\pm$, which contradicts extremality. Hence \(r\leq1\).

Suppose first that \(r=0\). Then \(c_{(1)}=\cdots=c_{(6)}=b\). Every active
six-subset consists of six entries equal to \(b\), so \(6b=1/2\) and
\(b=1/12\). Let \(m\) be the number of entries equal to \(b\). If \(m<12\),
normalization requires positive entries below \(b\) whose total is
\((12-m)/12\). Since each is strictly below \(1/12\), there must be at least
two of them. Adding and subtracting a small \(\varepsilon\) on two such
entries preserves normalization and every active six-sum, again contradicting
extremality. Thus \(m=12\), all other entries vanish, and the point is a
permutation of \(v(4)\).

Now suppose \(r=1\). Write the unique exceptional entry as \(a>b\).
Every active six-subset consists of that entry and five entries equal to \(b\),
hence
\begin{equation}\label{eq:active-ab}
 a+5b=\frac12.
\end{equation}
There can be no additional positive coefficient below \(b\). If there are at
least two, they may be perturbed against each other. If there is exactly one,
say \(c\), and \(m\) entries equal \(b\), the direction
\[
 \delta a=-5\varepsilon,\qquad
 \delta b=\varepsilon\ \text{on all \(m\) entries},\qquad
 \delta c=-(m-5)\varepsilon
\]
preserves both Eq.~\eqref{eq:active-ab} and normalization. For sufficiently small
positive and negative \(\varepsilon\), it also preserves all strict
inequalities and positivity. Thus every positive entry other than \(a\)
equals \(b\). There are \(n-1=15-z\) copies of \(b\). Together with normalization,
Eq.~\eqref{eq:active-ab} becomes
\[
 a+5b=\frac12,\qquad a+(15-z)b=1.
\]
Solving these equations gives Eq.~\eqref{eq:ab_app}. Since \(a>b\), this case has
\(z=0,1,2,3\).

It remains to prove that every listed point is genuinely a vertex. Let \( v(z)=\bigl(a_z,b_z,\ldots,b_z,0,\ldots,0\bigr) \) and suppose that it has a convex decomposition
\[
v(z)=t\,u+(1-t)w,\qquad 0<t<1,
\]
where \(u,w\in\mathcal C\). If \(v(z)_i=0\), then \( 0=t u_i+(1-t)w_i. \)
Since \(u_i,w_i\geq0\), this implies \(u_i=w_i=0\). Thus both \(u\) and \(w\) vanish at every zero coordinate of \(v(z)\). Therefore, although weight may be transferred from a positive coordinate to a zero coordinate in a one-sided feasible perturbation, such a perturbation cannot appear in a convex decomposition of \(v(z)\), because the other component would require a negative coefficient at that position.
Therefore, it remains to exclude the possibility for a nontrivial decomposition where \(u\) and \(w\) differ from \(v(z)\) only at its positive-coordinate positions. For every six-subset \(\bar S\) active at \(v(z)\), \( \sum_{i\in \bar S}v(z)_i=\frac12. \)
Using the convex decomposition,
\[
\frac12
=
t\sum_{i\in \bar S}u_i
+(1-t)\sum_{i\in \bar S}w_i.
\]
Since \(u,w\in\mathcal C\), both six-sums are at most \(1/2\). Therefore, both must equal \(1/2\). Hence \(u\) and \(w\) satisfy all the same active equalities as \(v(z)\).
For \(z\leq3\), subtracting two active equations that differ in only one repeated positive position forces all those coordinates of \(u\) to be equal. Thus
\(u=\left(\alpha,~\beta^{\times(15-z)},~0^{\times z}\right).\)
The active six-subset condition and normalization give
\[
\alpha+5\beta=\frac12,
\qquad
\alpha+(15-z)\beta=1.
\]
These equations have the unique solution \(\alpha=a_z\) and \(\beta=b_z\). Hence \(u=v(z)\), and similarly \(w=v(z)\).
For \(z=4\), the active equations force all twelve positive entries to be equal, and normalization fixes each one to \(1/12\). Thus again \(u=w=v(4)\).
Therefore, \(v(z)\) has no nontrivial convex decomposition and is a vertex of \(\mathcal C\).
\end{proof}

\section{Geometry-preserving symmetry groups}
\label{app:symmetry}
This section describes the geometry-preserving symmetry groups used in the two Bell-diagonal settings. For the two-qutrit Weyl basis, the relevant group is \(\operatorname{AGL}(2,3)\), acting on the affine plane \(\mathbb F_3^2\), whereas for the two-ququart lattice basis it is \(\operatorname{ASp}(4,2)\), acting on the binary symplectic space \(\mathbb F_2^4\). We specify the induced actions on the Bell coefficients, the resulting orbit structures, and why these transformations preserve separability.
\subsection{Affine symmetry in the Weyl--Heisenberg two-qutrit case}
\label{app:3x3_symmetry}
We identify the nine two-qutrit Bell-basis labels with the points of the affine plane \(\operatorname{AG}(2,3)\cong\mathbb{F}_3^2\)~\cite{Baumgartner2006,Baumgartner2007}.
Its affine general linear group is
\[
\operatorname{AGL}(2,3)
=
\mathbb{F}_3^2\rtimes\operatorname{GL}(2,3),
\]
whose elements act as
\begin{equation}
\label{eq:qutrit-affine-map}
g:x\longmapsto Ax+t,
\qquad
A\in\operatorname{GL}(2,3),\quad
t\in\mathbb{F}_3^2.
\end{equation}
The induced action on a coefficient vector \(c\) is given by the pullback
\begin{equation}
\label{eq:qutrit-affine-pullback}
(g\cdot c)_x
=
c_{g^{-1}(x)}
=
c_{A^{-1}(x-t)}.
\end{equation}
which maps the entry initially at position $x$ to the transformed coordinate $g(x)$. Since \( |\operatorname{GL}(2,3)|=48, \)
the affine group has order \( |\operatorname{AGL}(2,3)| = 9\cdot48 = 432. \)
Two coefficient vectors are called affine-equivalent when they belong to the same orbit
\begin{equation}
\label{eq:qutrit-affine-orbit}
\operatorname{Orb}(c) =\left\{g\cdot c:g\in\operatorname{AGL}(2,3)\right\}.
\end{equation}
The orbit--stabilizer theorem gives
\begin{equation}
 |\text{Orb}(c)|=\frac{432}{|\text{Stab}(c)|}, \qquad \text{Stab}(c)=\{g:g\cdot c=c\}.
\label{eq:orbitstabilizer}
\end{equation}

The determinant-one subgroup is the affine symplectic group
\[ \operatorname{ASp}(2,3) = \mathbb{F}_3^2\rtimes\operatorname{Sp}(2,3), \]
where \(\operatorname{Sp}(2,3)=\operatorname{SL}(2,3)\) in dimension two. Hence, in this case, $\operatorname{ASL}(2,3)=\operatorname{ASp}(2,3)$. Since $|\operatorname{Sp}(2,3)| =24$, affine symplectic group has order $|\operatorname{ASp}(2,3)|=9\cdot24=216.$

The action of this subgroup on the Bell projectors is implemented by local Weyl--Clifford unitaries~\cite{Hostens2005} and therefore preserves separability. The remaining elements of \(\operatorname{AGL}(2,3)\) have linear parts of determinant \(-1\). Let \( R:(k,l)\longmapsto(k,-l)\)
be the canonical reflection. If \(A\in\operatorname{GL}(2,3)\) satisfies
\(\det A=-1\), then
\[
S:=AR\in\operatorname{SL}(2,3),
\qquad
A=SR.
\]
Thus every determinant-minus-one transformation is a composition of a symplectic transformation and the canonical reflection \(R\). On the Bell projectors, this reflection is implemented by the subsystem-swap operator \(F\). Specifically, \( FP_{k,l}F^\dagger=P_{k,-l}. \) Although \(F\) is not a local unitary, it preserves the separability structure. Consequently, the entire \(\operatorname{AGL}(2,3)\) action, including its determinant-minus-one coset, preserves separability.

Geometrically, \(\operatorname{AGL}(2,3)\) preserves collinearity and therefore permutes the twelve affine lines of \(\operatorname{AG}(2,3)\). These lines form four parallel classes -- horizontal, vertical, diagonal, and antidiagonal -- with three lines in each class. Every line contains three points, and every two distinct points determine a unique affine line. Hence the affine action preserves precisely the incidence structure used in the qutrit separability construction. Under this action, the \(2725\) generating points in
\(\mathcal{V}_{2725}\) decompose into \(28\) affine orbits. It is important to note that affine inequivalence should not be confused with convex independence: a representative of one affine orbit may still belong to the convex hull of points from other orbits.

\subsection{Affine symplectic symmetry in the two-ququart lattice setting}
\label{sec:affinesymm_4x4}
The lattice analogue of the affine action introduced for the qutrit phase space is obtained by replacing the affine plane \(\mathbb{F}_3^2\) with the binary symplectic space \( V_4=\mathbb{F}_2^4. \) The relevant group is not the full affine general linear group, but the affine symplectic group
\begin{equation}
\operatorname{ASp}(4,2)
=
V_4\rtimes\operatorname{Sp}(4,2).
\end{equation}
Its elements act on lattice labels as
\begin{equation}
\label{eq:lattice-affine-action}
g:x\longmapsto Sx+t,
\qquad
S\in\operatorname{Sp}(4,2),\quad
t\in V_4,
\end{equation}
where \(S\) preserves the binary symplectic product:
\[
[Sx,Sy]=[x,y]
\qquad
\text{for all }x,y\in V_4.
\]

Using the same pullback convention as in the qutrit case, the induced action on a coefficient vector is
\begin{equation}
\label{eq:lattice-coefficient-action}
(g\cdot c)_x
=
c_{g^{-1}(x)}
=
c_{S^{-1}(x-t)}
=
c_{S^{-1}(x+t)},
\end{equation}
where the last equality follows because addition and subtraction coincide over \(\mathbb{F}_2\). 

The order of the linear symplectic group is \( |\operatorname{Sp}(4,2)|= 720,\) and therefore $|\operatorname{ASp}(4,2)|= 16\cdot720 = 11520$. Thus, with the orbit and stabilizer defined as in the qutrit discussion,
\[
|\operatorname{Orb}(c)|
=
\frac{11520}
{|\operatorname{Stab}(c)|}.
\]

The difference from the qutrit case lies in the allowed linear transformations. For the two-dimensional qutrit label space, every
\(A\in\operatorname{GL}(2,3)\) satisfies
\[
A^{T}J_2A=(\det A)J_2,
\]
where $J_{2n}=\bigoplus_{i=1}^n\begin{pmatrix}0&1\\-1&0\end{pmatrix}$ is the $2n\times2n$ matrix of the symplectic form, so that $[x,y]=x^TJ_{2n}y$. Since the only nonzero elements of \(\mathbb F_3\) are \(1\) and \(-1\), the determinant-one and determinant-minus-one sectors are respectively symplectic and anti-symplectic. This is why the full group $\AGL(2,3)$ preserves the structure of the qutrit affine Lagrangian lines, i.e. the cosets of Lagrangian subspaces. In contrast, for the four-dimensional binary space, the condition
\(\det S=1\) does not imply
\[
S^{T}J_4S=J_4.
\]
Indeed, every invertible matrix over \(\mathbb{F}_2\) has determinant one, whereas only the subgroup
\(\operatorname{Sp}(4,2)\subsetneq\operatorname{GL}(4,2)\) preserves the symplectic product. Moreover, because \(-J=J\) over \(\mathbb{F}_2\), there is no distinct anti-symplectic coset analogous to the determinant-minus-one part of the qutrit group.

The affine symplectic action preserves the lattice geometry relevant to separability. If \( L=x_0+H \) is an affine Lagrangian plane, then \( g(L)=Sx_0+t+S(H) \)
is again an affine Lagrangian plane, because \(S(H)\) is Lagrangian whenever \(H\) is Lagrangian. Consequently, the state associated to the affine Lagrangian plane $L$ transforms as
\begin{equation}
\label{eq:lattice-plane-action}
\sigma_L = \frac14\sum_{x\in L}P_x \quad\longmapsto\quad \sigma_{g(L)}.
\end{equation}

This action is implemented by local Clifford unitaries. For every
\(S\in\operatorname{Sp}(4,2)\), there exists a two-qubit Clifford unitary \(C_S\) satisfying~\cite{Dehaene2003,Hostens2005}
\[
C_SQ_xC_S^\dagger
=
e^{i\theta_S(x)}Q_{Sx}.
\]
For \(g(x)=Sx+t\), define
\[
W_g=C_S^*\otimes Q_tC_S,
\]
where the two tensor factors act on the two bulk parties. Then
\begin{align}
W_g\ket{\Omega_x}
&=
\left(C_S^*\otimes Q_tC_S\right)
\left(I_4\otimes Q_x\right)\ket{\Omega_4}
\nonumber\\
&=
e^{i\phi_g(x)}
\ket{\Omega_{Sx+t}},
\end{align}
and hence
\begin{equation}
\label{eq:lattice-projector-action}
W_gP_xW_g^\dagger=P_{g(x)},
\qquad
W_g\rho(c)W_g^\dagger=\rho(g\cdot c).
\end{equation}
Since \(W_g\) is a product unitary with respect to the bipartition, it preserves separability.

Finally, two different symmetry levels should be distinguished. The \(\PPPT\) polytope \(\mathcal C\), being defined solely through the six largest coefficients, is invariant under the full coordinate-permutation group \(S_{16}\). The subgroup
\(\operatorname{ASp}(4,2)\subset S_{16}\), however, is the relevant geometric symmetry because it additionally preserves the symplectic product, the affine Lagrangian planes, and the associated separable plane states. A general element of \(\operatorname{AGL}(4,2)\) need not preserve this Lagrangian structure.

\begin{remark}[The two ququart Bell bases]
\label{rem:ququart-bases}
The standard ququart Weyl basis is labeled by \(\mathbb Z_4^2\) and generated by the operators \(X_4^kZ_4^l\), whereas the lattice Bell basis is labeled by \(\mathbb F_2^4\) and generated by two-qubit Pauli operators. These two maximally entangled bases are not locally unitarily equivalent.

Indeed, suppose that a product unitary \(A\otimes B\) mapped the Weyl basis to the lattice basis, up to phases and a permutation \(\pi\in S_{16}\). Then, for every pair of Weyl labels \(u,v\), the corresponding operator quotients would satisfy
\[
B\bigl(W_uW_v^\dagger\bigr)B^\dagger \propto Q_{\pi(u)}Q_{\pi(v)}^\dagger.
\]
Unitary conjugation preserves projective order\footnote{Projective order of a unitary $U$ is the minimal integer $r$ such that $U^r=e^{\iota \theta}I$ for some phase $e^{\iota \theta}$}. However, some nonidentity Weyl quotients, such as \(X_4\), have projective order four, whereas every nonidentity quotient of two lattice Pauli operators has
projective order two. This contradiction proves that no product unitary maps one complete Bell basis to the other, even after allowing phases and permutations. A global basis-changing unitary exists, but it does not preserve the local geometry used in our analysis.
\end{remark}

\section{The fifteen ordered vertices of the polytope $Q_9^\downarrow$}
\label{app:ordered_vertices}
Below are the fifteen ordered vertices of the outer polytope $Q_9^\downarrow$, listed in sequence.
\begin{align}
V_1&=(1/9)^{\times9},\nonumber\\
V_2&=(1/8)^{\times8},0,\nonumber\\
V_3&=(1/7)^{\times7},0,0,\nonumber\\
V_4&=(1/7)^{\times6},(1/21)^{\times3},\nonumber\\
V_5&=(1/7)^{\times6},(1/14)^{\times2},0,\nonumber\\
V_6&=(3/19)^{\times5},(1/19)^{\times4},\nonumber\\
V_7&=(2/13)^{\times5},(1/13)^{\times3},0,\nonumber\\
V_8&=(3/17)^{\times4},(1/17)^{\times5},\nonumber\\
V_9&=(1/6)^{\times4},(1/12)^{\times4},0,\nonumber\\
V_{10}&=(1/5)^{\times3},(1/15)^{\times6},\nonumber\\
V_{11}&=(2/11)^{\times3},(1/11)^{\times5},0,\nonumber\\
V_{12}&=(3/13)^{\times2},(1/13)^{\times7},\nonumber\\
V_{13}&=(1/5)^{\times2},(1/10)^{\times6},0,\nonumber\\
V_{14}&=3/11,(1/11)^{\times8},\nonumber\\
V_{15}&=2/9,(1/9)^{\times7},0.
\label{eq:15vertices}
\end{align}
\section{Representative counting arguments under $\AGL(2,3)$ symmetry}
\label{app:rep_count_3x3}
Following our earlier convention, we label the distinct entries within a given spectral profile by $H$, $M$, $L$, and $Z$, denoting high, intermediate, low, and zero values (when present), respectively. Any configuration is uniquely specified by the phase-space coordinates supporting the non-generic values; the remaining positions are then occupied automatically. For brevity, we focus our analysis on the non-trivial spectral types where orbit splitting is not immediate.
\paragraph{Type 3.}
For $H^{\times7}Z^{\times2}$, choosing the two zero positions completely fixes the state. Any two zeros determine a unique affine line, whose third point is necessarily an $H$. Thus the configurations may be viewed as lines with coefficient pattern $HZZ$. A line with coefficient pattern $HZZ$ means only that the three points on that line carry one $H$ and two $Z$'s; their placement along the line is irrelevant. There are twelve geometric lines and three choices for which point of the line carries $H$, so
\begin{equation}
 12\times3=36.
\end{equation}
All such configurations are affine-equivalent, hence type 3 is one orbit of size 36.

\paragraph{Type 4.}
For $H^{\times6}L^{\times3}$, the three identical $L$ positions determine the configuration. There are $\binom93=84$ triples. Exactly twelve are affine lines; the remaining $84-12=72$ are noncollinear. Since $\AGL(2,3)$ is transitive on lines and on noncollinear triples,
\begin{equation}
 84=12+72.
\end{equation}
Here the line count is twelve, rather than $9\times4=36$, because $9\times4$ counts each line passing through all $L$ three times, once from each of its identical $L$ points.

\paragraph{Type 5.}
For $H^{\times6}M^{\times2}Z$, the three marked points have coefficient patterns $ZMM$. In the collinear case, choose $Z$ in nine ways and then one of the four lines through it; the two remaining points of that line are forced to be the identical $M$'s. Hence
\begin{equation}
 N_{\mathrm{coll}}=9\times4=36.
\end{equation}
Equivalently, one may choose $Z$ in nine ways and one of the eight positions for the first $M$, after which the second $M$ is fixed; since two $M$s are indistinguishable, division by $2!$ gives $9\times8/2=36$. The remaining $252-36=216$ configurations are noncollinear.

\paragraph{Type 6.}
The four distinguished $L$ points either contain a line or contain no collinear triple. In the first case choose the line in 12 ways and the fourth point outside it in 6 ways, giving
\begin{equation}
 12\times6=72.
\end{equation}
The remaining $\binom94-72=54$ four-point subsets are four-caps (four-point subset containing no three collinear points).

\paragraph{Type 7.}
Write the marked points as $L_1,L_2,L_3,Z$. There are $9\binom83=504$ placements. Three affine cases occur. First, $L_1,L_2,L_3$ themselves form a line: choose that line in 12 ways and choose $Z$ outside it in 6 ways, giving 72. Second, the $L$-triple is noncollinear and $Z$ lies on one of its three secants. There are $84-12=72$ noncollinear $L$-triples, and
each has three such positions for $Z$, giving $72\times3=216$. Third, no three of the four marked points are collinear; for each noncollinear $L$-triple, the other three available positions (i.e., positions other than its three secants) of $Z$ give $72\times3=216$. Thus
\begin{equation}
 504=72+216+216.
\end{equation}

\paragraph{Type 9.}
For $H^{\times4}L^{\times4}Z$, it is convenient to mark only $H_1,H_2,H_3,H_4,Z$; the four $L$ positions are then fixed. A marked line can only have pattern $HHH$ or $ZHH$. Four incidence classes occur:
\begin{align}
&\text{exactly one }HHH\text{ line}: &&144,\\
&\text{exactly one }ZHH\text{ line}: &&216,\\
&\text{one }HHH\text{ and one }ZHH\text{ line}: &&216,\\
&\text{two }ZHH\text{ lines}: &&54.
\end{align}
For example, the first number follows by choosing an $HHH$ line (12 choices), choosing the fourth $H$ outside it (6 choices), and then choosing one of the two remaining positions for $Z$ that creates no additional marked line:
\begin{equation}
 12\times6\times2=144.
\end{equation}
For the last class, fix $Z$. The eight remaining points are partitioned into four pairs by the four lines through $Z$. Choosing two of those pairs to be occupied by the four $H$ points gives
\begin{equation}
 9\binom42=54.
\end{equation}
The four multiplicities sum to
\begin{equation}
 144+216+216+54=630.
\end{equation}

The remaining rows follow from the same incidence logic, often by exchanging a coefficient class with its complement. No additional affine classes occur.

\section{The 28 affine-orbit representatives under $\AGL(2,3)$ symmetry}
\label{app:orbit_representative_3x3}

Below, we list the $28$ orbit representatives used to construct the exact separability certificate, ordered according to the phase-space coordinates $(00, 01, 02, 10, 11, 12, 20, 21, 22)$:
\begin{align*}
T_1&=(1/9,1/9,1/9,1/9,1/9,1/9,1/9,1/9,1/9),\\
T_2&=(0,1/8,1/8,1/8,1/8,1/8,1/8,1/8,1/8),\\
T_3&=(0,0,1/7,1/7,1/7,1/7,1/7,1/7,1/7),\\
T_4&=(1/21,1/21,1/21,1/7,1/7,1/7,1/7,1/7,1/7),\\
T_5&=(1/21,1/21,1/7,1/21,1/7,1/7,1/7,1/7,1/7),\\
T_6&=(0,1/14,1/14,1/7,1/7,1/7,1/7,1/7,1/7),\\
T_7&=(0,1/14,1/7,1/14,1/7,1/7,1/7,1/7,1/7),\\
T_8&=(1/19,1/19,1/19,1/19,3/19,3/19,3/19,3/19,3/19),\\
T_9&=(1/19,1/19,3/19,1/19,1/19,3/19,3/19,3/19,3/19),\\
T_{10}&=(0,1/13,1/13,1/13,2/13,2/13,2/13,2/13,2/13),\\
T_{11}&=(0,1/13,2/13,1/13,1/13,2/13,2/13,2/13,2/13),\\
T_{12}&=(0,1/13,2/13,1/13,2/13,2/13,2/13,2/13,1/13),\\
T_{13}&=(1/17,1/17,1/17,1/17,1/17,3/17,3/17,3/17,3/17),\\
T_{14}&=(1/17,1/17,1/17,1/17,3/17,3/17,1/17,3/17,3/17),\\
T_{15}&=(0,1/12,1/12,1/12,1/12,1/6,1/6,1/6,1/6),\\
T_{16}&=(0,1/12,1/12,1/12,1/6,1/6,1/12,1/6,1/6),\\
T_{17}&=(0,1/12,1/12,1/12,1/6,1/6,1/6,1/12,1/6),\\
T_{18}&=(0,1/12,1/6,1/12,1/12,1/12,1/6,1/6,1/6),\\
T_{19}&=(1/15,1/15,1/15,1/15,1/15,1/15,1/5,1/5,1/5),\\
T_{20}&=(1/15,1/15,1/15,1/15,1/15,1/5,1/15,1/5,1/5),\\
T_{21}&=(0,1/11,1/11,1/11,1/11,1/11,2/11,2/11,2/11),\\
T_{22}&=(0,1/11,1/11,1/11,1/11,2/11,1/11,2/11,2/11),\\
T_{23}&=(0,1/11,1/11,1/11,1/11,2/11,2/11,1/11,2/11),\\
T_{24}&=(1/13,1/13,1/13,1/13,1/13,1/13,1/13,3/13,3/13),\\
T_{25}&=(0,1/10,1/10,1/10,1/10,1/10,1/10,1/5,1/5),\\
T_{26}&=(0,1/10,1/10,1/10,1/10,1/5,1/10,1/5,1/10),\\
T_{27}&=(1/11,1/11,1/11,1/11,1/11,1/11,1/11,1/11,3/11),\\
T_{28}&=(0,1/9,1/9,1/9,1/9,1/9,1/9,1/9,2/9).
\end{align*}

\section{Derivation of the thirteen lattice orbits under $\ASp(4,2)$ symmetry}
\label{app:lattice-orbits_4x4}

For \(z\leq3\), a lattice vertex is determined by a distinguished position
\(d\in\mathbb F_2^4\) carrying \(a_z\) and an unordered set \(Z\subset\mathbb F_2^4\setminus\{d\}\) of $z$ zero positions, with \(d\notin Z\).
For \(g(x)=Sx+t\), choose \(t=Sd\).  Because subtraction equals addition in
\(\F_2\), \(g(d)=0\) and \(g(x)=S(x+d)\).
Thus the classification reduces to the action of \(\Sp(4,2)\) on the
translated zero set \(Z+d\).  Symplectic products of differences are
invariant:
\begin{equation}
 [g(x)+g(d),g(y)+g(d)]=[x+d,y+d].
\end{equation}
Witt's extension lemma ensures that the isometries between the small
subspaces used below extend to global symplectic transformations.

\subsection*{The cases \(z=0\) and \(z=1\)}

For \(z=0\), translations are transitive on the distinguished point, giving
one orbit of size \(16\).  For \(z=1\), after moving \(d\) to zero, the unique zero coefficient is located at some nonzero label \(u\in\mathbb F_2^4\).  The group \(\Sp(4,2)\) is transitive on the
fifteen nonzero vectors, giving one orbit of size \(16\cdot15=240\).

\subsection*{The case \(z=2\)}

Let us consider the translated zeros as \(\{u,v\}\).  The invariant \([u,v]\) gives
the two orbits
\begin{equation}\label{eq:z2-reps}
 [u,v]=1:\ \{\hex{1},\hex{2}\},
 \qquad
 [u,v]=0:\ \{\hex{1},\hex{4}\}.
\end{equation}
For each of the fifteen non-zero vectors $u \in \mathbb{F}_2^4$, there exist exactly eight non-zero vectors whose symplectic inner product with $u$ equals $1$, and six distinct non-zero vectors whose symplectic inner product with $u$ vanishes. Consequently, there are $\frac{15 \cdot 8}{2} = 60$ non-commuting pairs and $\frac{15 \cdot 6}{2} = 45$ commuting pairs, where the factor of $1/2$ accounts for the symmetry of unordered pairs (since both $u$ and $v$ support vanishing coefficients, their permutation yields the same vertex).  Restoring \(d\) gives orbit sizes \(16\cdot60=960\) and \(16\cdot45=720\).

\subsection*{The case \(z=3\)}

Let the translated zero set be \(\{u,v,w\}\).  Let
\begin{eqnarray}
\label{eq:z3-invariants}
 r&=&\dim(\operatorname{span}\{u,v,w\}),\;\;\text{and}\\
 m&=&\#\{\{x,y\}\subset\{u,v,w\}:[x,y]=1\}.
\end{eqnarray}
If \(r=2\), the set consists of all three nonzero vectors of a two-dimensional subspace.  That subspace is either nondegenerate (\(m=3\)) or Lagrangian (\(m=0\)), represented by
\begin{equation}\label{eq:z3-rank2}
 \{\hex{1},\hex{2},\hex{3}\},
 \qquad
 \{\hex{1},\hex{4},\hex{5}\}.
\end{equation}
There are twenty nondegenerate and fifteen Lagrangian two-dimensional subspaces, so the orbit sizes after restoring \(d\) are \(320\) and \(240\).

If \(r=3\), the three vectors are linearly independent. The case \(m=0\) is impossible, since it would make their span a three-dimensional totally isotropic subspace (Lagrangian space), whereas every totally isotropic subspace of a nondegenerate four-dimensional symplectic space has dimension at most two. Up to reordering \(u,v,w\), the three possible pairing patterns are distinguished by whether exactly one, two, or all three of the symplectic products \([u,v]\), \([u,w]\), and \([v,w]\) equal \(1\). These cases correspond to \(m=1,2,3\) and are represented, respectively, by
\begin{equation}\label{eq:z3-rank3}
 \{\hex{1},\hex{2},\hex{4}\},\quad
 \{\hex{1},\hex{2},\hex{5}\},\quad
 \{\hex{1},\hex{2},\hex{7}\}.
\end{equation}
The triples in each case can be counted as follows. For \(m=1\), choose the unique pair with symplectic product \(1\), giving \(60\) choices. There are then \(3\) choices for a third vector that is orthogonal to both and makes the triple linearly independent. This gives \(60\cdot3=180\) triples. For \(m=2\), choose the unique pair with symplectic product \(0\), giving \(45\) choices. There are \(4\) choices for a third vector whose symplectic product with each of the first two vectors equals \(1\), giving \(45\cdot4=180\) triples. For \(m=3\), choose a pair with symplectic product \(1\), then choose one of the \(3\) third vectors that make the triple linearly independent and have symplectic product \(1\) with each of the first two vectors. Each triple is counted three times, since any of its three pairs can be chosen as the initial pair, so this gives \(60\cdot3/3=60\) triples. Including the \(16\) possible distinguished labels therefore gives affine symplectic orbit sizes \(2880,2880,\) and \(960\), respectively.

\subsection*{The case \(z=4\)}

For \(z=4\), we have \(a_4=b_4=1/12\), so there is no distinguished label. Each vertex is therefore determined solely by the unordered set \(Z={x_1,x_2,x_3,x_4}\) of labels at which its coefficients vanish.  Choose \(x_1\) as an origin and define
\begin{equation}
\label{eq:direction-space}
 W=\operatorname{span}\{0,x_2+x_1=u,x_3+x_1=v,x_4+x_1=w\}.
\end{equation}
Translations do not affect these differences. The affine dimension is two or three.

If \(\dim W=2\), then \(Z=x_1+W\) is a two-dimensional affine plane whose direction space \(W\) is either nondegenerate or Lagrangian.  There are twenty nondegenerate and fifteen Lagrangian two-dimensional subspaces, each with four distinct cosets. This gives
orbits of sizes \(80\) and \(60\), represented by
\begin{equation}\label{eq:z4-plane-reps}
 \{\hex{0},\hex{1},\hex{2},\hex{3}\},
 \qquad
 \{\hex{0},\hex{1},\hex{4},\hex{5}\}.
\end{equation}

Suppose now that \(\dim W=3\). The restriction of the alternating bilinear form \([\cdot,\cdot]\) to \(W\) has a one-dimensional radical,
\begin{equation}
\label{eq:radW}
R=\rad(W)=W\cap W^\perp=\{0,r\}.
\end{equation}
Define \( \Delta=u+v+w=x_1+x_2+x_3+x_4.\) This vector is invariant under translations and depends only on the unordered set \(Z\), independently of its ordering or the choice of reference point (here it is \(x_1\)). Under a symplectic transformation \(S\), both \(\Delta\) and \(r\) transform to \(S\Delta\) and \(Sr\), respectively. Thus the two cases \(\Delta=r\) and \(\Delta\ne r\) are preserved by affine symplectic transformations. Although a pairwise sum such as \(u+v\) also transforms covariantly, it selects a particular pair, so the condition \(u+v=r\) depends on the ordering of the same set. Here ``ordering'' means the arbitrary labeling of the three direction vectors as \(u,v,w\). Relabeling them, equivalently, changing their order, may replace \(u+v\) by \(u+w\) or \(v+w\), so the condition \(u+v=r\) may no longer hold even though the underlying unordered set \(Z\) is unchanged. The geometric distinction between the two cases follows from the six pairwise differences of the points in \(Z\): \( u,\ v,\ w,\ u+v,\ u+w,\ v+w. \) Since \(u,v,w\) are linearly independent, these are precisely the six nonzero vectors of \(W\) other than \(\Delta\). Consequently, the radical vector occurs as a pairwise difference exactly when \(\Delta\neq r\).

If \(\Delta\ne r\), there exist distinct points \(x,y\in Z\) such that \(x+y=r\). Applying the translation \(z\mapsto z+x\) to every point of \(Z\) sends \(x\) to \(0\) and \(y\) to \(r\). Denoting the other two translated points by \(a,b\), the set becomes \(\{0,r,a,b\}\). Since \(r\) lies in the radical, \([r,a]=[r,b]=0\). If \([a,b]=0\) as well, the form would vanish on \(W=\operatorname{span}\{r,a,b\}\), contradicting its one-dimensional radical. Hence \([a,b]=1\). The set $\{0,r,a,b\}$ with $[a,b]=1$ is equivalent, under the affine symplectic group, to $\{\hex{0},\hex{1},\hex{2},\hex{4}\}$.  If \(\Delta=r\),
the set is equivalent to \(\{0,a,b,r+a+b\}\), represented by
\(\{\hex{0},\hex{1},\hex{2},\hex{7}\}\). Since each three-dimensional subspace $W$ has the form $W=r^\perp$ with $r\neq0$, and there are fifteen nonzero vectors in $\mathbb F_2^4$, there are exactly fifteen such three-dimensional subspaces. Each three-dimensional subspace has two affine cosets. Fix a three-dimensional direction space \(W\) and one of its affine cosets \(C=t+W\). For any prescribed nonzero \(\delta\in W\), a four-point subset \(Z\subset C\) of affine dimension three satisfying \(\Delta=\delta\) can be written as
\[
Z=\{x_1,x_1+u,x_1+v,x_1+w\},
\;\; w=\delta+u+v.
\]
There are \(8\) choices for \(x_1\) (since \(x_1\) is an element of coset $C$, where $C$ has dimension three, i.e., it has eight elements), \(6\) choices for \(u\notin\{0,\delta\}\), and \(4\) choices for \(v\notin\operatorname{span}\{u,\delta\}\). Dividing the resulting \(8\cdot6\cdot4\) ordered descriptions by \(4!\) gives \(8\) unordered subsets for each fixed \(\delta\neq0\). Since \(W\) has seven nonzero vectors, exactly one of which is its radical vector \(r\), each coset contains \(8\) subsets with \(\Delta=r\) and \(6\cdot8=48\) subsets with \(\Delta\neq r\). Hence, their orbit sizes are \(15\cdot2\cdot48=1440\) and \(15\cdot2\cdot8=240\).

\section{Computational algorithm for systematically extracting the affine orbits}
\label{app:numerical_algorithm}
\textit{\textbf{Two-qutrit case.}} All affine orbits, along with their multiplicities and canonical representatives, can be determined numerically via a straightforward exhaustive orbit-stripping routine. Let $\{V_r\}_{r=1}^{15}$ denote the extreme points of $Q_9^\downarrow$. For each generator $V_r$, construct the set $S_r = \{ P_\pi V_r : \pi \in S_9 \}$ of all its distinct coordinate permutations. Initializing a remaining-vector set $R \gets S_r$, we select an arbitrary element $c \in R$ and generate its complete affine orbit,
\begin{equation*}
\operatorname{Orb}(c) = { g \cdot c : g \in \mathrm{AGL}(2,3) }.
\end{equation*}
We record $c$ as an orbit representative and its cardinality $\vert{}\operatorname{Orb}(c)\vert{}$ as the corresponding multiplicity. The remaining-vector set is then updated via $R \gets R \setminus \operatorname{Orb}(c)$, and this whole procedure is repeated until $R = \emptyset$. If this process terminates after $k_r$ iterations, then $V_r$ decomposes into exactly $k_r$ inequivalent $\mathrm{AGL}(2,3)$ orbits. Iterating over all $r \in \{1, \dots, 15\}$ yields the complete orbit decomposition across $\mathcal{V}_{2725}$, as summarized in Algorithm~1.\\
\noindent
\rule{\columnwidth}{0.5pt}
\textbf{Algorithm 1.} Decomposition of the fifteen ordered spectral types $\{V_r\}_{r=1}^{15}$ into $28$ inequivalent $\AGL(2,3)$ orbits.
\medskip
\begin{algorithmic}[1]
\Require Ordered vertices $\{V_r\}_{r=1}^{15}$ of $Q_9^\downarrow$
\Ensure Affine-orbit representatives and their multiplicities
\State $\mathcal{R}_{\mathrm{all}}\gets\emptyset$
\For{$r=1,\ldots,15$}
    \State $S_r\gets
    \{\text{all distinct permutations of }V_r\}$
    \State $R\gets S_r$
    \State $\mathcal{R}_r\gets\emptyset$
    \While{$R\neq\emptyset$}
        \State Choose any $c\in R$
        \State $\mathcal{O}(c)\gets
        \{g\cdot c:\ g\in\AGL(2,3)\}$
        \State $m(c)\gets|\mathcal{O}(c)|$
        \State Record $(c,m(c))$ in $\mathcal{R}_r$
        \State $R\gets R\setminus\mathcal{O}(c)$
    \EndWhile
    \State $\mathcal{R}_{\mathrm{all}}
    \gets
    \mathcal{R}_{\mathrm{all}}\cup\mathcal{R}_r$
\EndFor
\State \Return $\mathcal{R}_{\mathrm{all}}$
\end{algorithmic}
\noindent
\rule{\columnwidth}{0.5pt}
We remark that the representative $c$ selected for each orbit $\operatorname{Orb}(c)$ is not unique; any element $c' \in \operatorname{Orb}(c)$ is an equally valid choice. Furthermore, identifying ``28 representatives'' signifies strictly that the $2725$ coefficient vectors partition into 28 disjoint equivalence classes under $\AGL(2,3)$; it does not assert that these 28 states are convexly independent.

\textit{\textbf{Two-ququart case.}} The orbit-extraction algorithm described for the two-qutrit case applies directly to the two-ququart lattice case, with  \(\operatorname{ASp}(4,2)\) replacing \(\operatorname{AGL}(2,3)\). For each of the five ordered spectral types \(v(z)\), \(z=0,1,2,3,4\), we generate all distinct permutations of its sixteen coefficients. We then select an unclassified vertex, compute its complete orbit under \(\operatorname{ASp}(4,2)\), record one representative and the corresponding orbit size, and remove the entire orbit from the remaining vertices. Repeating this procedure until all vertices are classified produces respectively \((1,1,2,5,4)\) orbits for the five spectral types. Thus, the \(11036\) vertices decompose into exactly thirteen affine-symplectic orbits, whose representatives and multiplicities are listed in Table~\ref{tab:lattice-orbits}.
\section*{Acknowledgment}
The authors used OpenAI's ChatGPT (GPT-5.6 Sol) to verify selected arguments and steps of the logic, and to paraphrase parts of the manuscript for better readability. Claude Sonnet 5.5 was used to improve the style and readability of the \textit{Wolfram Mathematica} notebooks. The authors take full responsibility for all claims, calculations, and conclusions. 
\bibliographystyle{apsrev4-2}
\bibliography{mainbib}
\end{document}